\documentclass[a4paper, authoryear, 11pt]{elsarticle}

\usepackage[english]{babel}
\usepackage{color}
\usepackage{fontenc}
\usepackage{graphicx}
\usepackage{epstopdf}
\usepackage{amsmath,amsthm,amssymb,bm}
\usepackage[colorlinks=true,citecolor=blue]{hyperref}
\usepackage[mathlines]{lineno}
\usepackage{tikz}
\usepackage[autostyle=true]{csquotes}%%for quotation
\MakeOuterQuote{"}
\usepackage{grffile}
\usepackage{enumerate}
\usepackage[useregional]{datetime2}
\usepackage{multirow} 
\usepackage{mathpazo}
\usepackage{chemformula}
\usepackage{arydshln}%%for dash line in matrix
\usepackage{mathrsfs}

\NewDocumentCommand{\evalat}{sO{\big}mm}{%
	\IfBooleanTF{#1}
	{\mleft. #3 \mright|_{#4}}
	{#3#2|_{#4}}%
}%% evaluated at
\usepackage{float}
\restylefloat{table}
\usepackage{commath}               %%For mathematical operator (derivative & other)
\numberwithin{equation}{section}
\newtheorem{theorem}{Theorem}[section]

\input epsf.sty
\usepackage{lineno}
\journal{arXiv}

\begin{document}
	\title{A mechanistic model for airborne and direct human-to-human transmission of COVID-19: Effect of mitigation strategies and immigration of infectious persons}
	
\author[mymainaddress]{Saheb Pal}\ead{saheb92.math@gmail.com}
\author[mysecondaryaddress]{Indrajit Ghosh\corref{mycorrespondingauthor}}\cortext[mycorrespondingauthor]{Corresponding author}\ead{indra7math@gmail.com}
\address[mymainaddress]{Department of Mathematics, Visva-Bharati, Santiniketan-731235, India}
\address[mysecondaryaddress]{Department of Computational and Data Sciences, Indian Institute of Science,\\ Bangalore - 560012, Karnataka, India} 
\begin{abstract}
The COVID-19 pandemic is the most significant global crisis since World War II that affected almost all the countries of our planet. To control the COVID-19 pandemic outbreak, it is necessary to understand how the virus is transmitted to a susceptible individual and eventually spread in the community. The primary transmission pathway of COVID-19 is human-to-human transmission through infectious droplets. However, a recent study by Greenhalgh et al. (Lancet: 397:1603-1605, 2021) demonstrates 10 scientific reasons behind the airborne transmission of SARS-COV-2. In the present study, we introduce a novel mathematical model of COVID-19 that considers the transmission of free viruses in the air beside the transmission of direct contact with an infected person. The basic reproduction number of the epidemic model is calculated using the next-generation operator method and observed that it depends on both the transmission rate of direct contact and free virus contact. The local and global stability of disease-free equilibrium (DFE) is well established. Analytically it is found that there is a forward bifurcation between the DFE and an endemic equilibrium using central manifold theory. Next, we used the nonlinear least-squares technique to identify the best-fitted parameter values in the model from the observed COVID-19 mortality data of two major districts of India. Using estimated parameters for Bangalore urban and Chennai, different control scenarios for mitigation of the disease are investigated. Results indicate that the vaccination of susceptible individuals and treatment of hospitalized patients are very crucial to curtail the disease in the two locations. It is also found that when a vaccine crisis is there, the public health authorities should prefer to vaccinate the susceptible people with compared to the recovered persons who are now healthy. Along with face mask use, treatment of hospitalized patients and vaccination of susceptibles, immigration should be allowed in a supervised manner so that economy of the overall society remain healthy.

\end{abstract}


%\linenumbers

\maketitle
\section{Introduction}\label{Sec:Intro}
The current outbreak of coronavirus disease 2019 (COVID-19) has significantly affected public health and the economy worldwide. As of 31 May 2021, more than 171 million COVID-19 cases and 3 million deaths have been reported globally \cite{Worldometer2021covid}. COVID-19 is an emerging respiratory infectious disease that was first detected in early December 2019 in Wuhan, China. The virus that causes COVID-19 is severe acute respiratory syndrome coronavirus 2 (SARS-CoV-2), which spread rapidly throughout the globe. The World Health Organization (WHO) initially declared the COVID-19 outbreak as a public health emergency of international concern but with the spread across and within countries leading to a rapid increase in the number of cases, the outbreak was officially declared a global pandemic on 11 March 2020 \cite{who2020rolling}. Given the global scenario on the series of waves and unprecedented strains of COVID-19, there is a need for more investigation to timely and effectively curtail the spread of the disease. 

Mathematical modelling is a very versatile and effective instrument for understanding infectious disease transmission dynamics and helps us to develop preventive measures for controlling the infection spread both qualitatively and quantitatively. To model the epidemiological data of infectious diseases, compartmental models are widely used due to their simplicity and well-studied applications \cite{kermack1927contribution}. In this regard, the recent epidemiological threat of COVID-19 is not an exception. In such types of models, the population is categorized into several classes, or compartments, based on the stage of the infection that is affecting them. These models are governed by a system of ordinary differential equations that take into account the time-discoursed infection status of the population compartments. Wu et al. \cite{wu2020nowcasting} first developed a mathematical model for COVID-19 which is based on SEIR (Susceptible-Exposed-Infectious-Recovered) model and explains the transmission dynamics and estimate the national and global spread of COVID-19. After that, various mathematical models were proposed by considering several stages of susceptible and infected populations \cite{yan2020impact, mwalili2020seir, choi2020optimal, senapati2021impact}. Many researchers studied the dynamics of COVID-19  using real incidence data of the affected countries and examined different characteristics of the outbreak as well as evaluated the effect of intervention strategies implemented to suppress the outbreak in respective countries \cite{yan2020impact,choi2020optimal}. 

Presently, the second wave of COVID-19 is affecting most of the countries in the world. In a densely populated country like India, this scenario is more threatening. On April 15, 2021, the daily cases of COVID-19 were double the first peak. However, The epidemic evolution of the country is quite complex due to regional inhomogeneities. Thus, governments of various states in India are adopting different strategies to control the outbreak, such as wearing masks, vaccination drives, social distancing guidelines, partial lockdowns, and restricted store hours in place. Thus, gaining an understanding of how this outbreak spread in the form of waves and possible interventions on a short-term basis is urgent. As of 31 May 2021, more than 28 million COVID-19 cases and more than 331 thousand deaths are reported in India \cite{covid19india}. We have chosen to concentrate on two major districts namely, Bangalore urban and Chennai of India for our case studies. These two districts are very crucial hubs of south India and are experiencing a severe second wave of COVID-19. Additionally, taking small regions is also necessary to make reliable projections using mechanistic mathematical models. Due to the small area and populations, it is also easy to implement and supervise specific control measures.

As a case study in India, several researchers proposed models that are fitted to the daily COVID-19 cases and deaths, and examined different control strategies \cite{mandal2020prudent, sardar2020assessment, rafiq2020evaluation, sarkar2020modeling, senapati2021impact, ghosh2021modeling}. Some of the studies explore the vaccine allocation strategy in India due to limited supply and to support relevant policies \cite{foy2021comparing}. However, none of the above studies has considered the airborne transmission of the SARS-COV-2 virus in their respective models with application to Indian districts.

To control a pandemic outbreak, it is necessary to understand how the virus transmits to a susceptible individual and eventually spreads in the community. The primary transmission pathway of COVID-19 is human-to-human propagation through infectious droplets \cite{van2020aerosol}. However, a recent study of Greenhalgh et al. \cite{greenhalgh2021ten} suggests that SARS-COV-2 can be transmitted through the air and they showed 10 scientific reasons behind this. Later on, Addleman et al. \cite{addleman2021mitigating} commented that Canadian public health guidance and policies should be updated to address the airborne mode of transmission. The authors also suggested that addressing airborne transmission requires the expertise of interdisciplinary teams to inform solutions that can end this pandemic faster. Motivated by these scientific pieces of evidence, we consider a compartmental model of SEIR-type including shedding of free virus in the air by infectious persons. To the best of our knowledge, this is the first modelling study that considers the airborne transmission of COVID-19 with applications to Indian districts. We consider that the susceptible population becomes exposed in two different ways: firstly through direct contact with an infected population or touches a surface that has been contaminated. This transmission also occurs through large and small respiratory droplets that contain the virus, which would occur when near an infected person. The next is through the airborne transmission of smaller droplets and particles that are suspended in the air over longer distances and time \cite{editorial2020covid}. We also study the effect of immigration along with popular control strategies. The main focus of our study is to explore the following epidemiological issues:
\begin{itemize}
	\item {Mathematically analyze the COVID-19 transmission dynamics by incorporating the airborne pathway of free SARS-CoV-2 virus into the model structure.}
	
	\item {Impact of anti-COVID drugs on the reduction of hospitalized persons.}
	
	\item {Individual and combined effects of various control measures (use of face mask, vaccination) as well as immigration of infectious persons on the COVID-19 pandemic.}	
\end{itemize}

The remainder of this paper is organized as follows. In the next section, we propose a deterministic compartmental model to describe the disease transmission mechanism. We consider the amount of free virus in the air as a dynamic variable. Section \ref{sec:math} describes the theoretical analysis of the model, which incorporates the existence of positive invariance region, boundedness of solutions, computation of the basic reproduction number, and stability of disease-free equilibrium. Also, in this section, the existence of forward bifurcation of the model system is explored. In Section \ref{Sec:calibration}, we fit the mathematical model using nonlinear least-squares technique from the observed mortality data and estimate unknown parameters. Section \ref{Sec:controls} describe several control mechanisms and immigration of infectives through numerical simulation. We examine the effects of vaccination, treatment by drugs, and use of face masks with different degrees of efficacy as intervention strategies. Finally, in Section \ref{conclusion} we discuss the findings and some concluding remarks about mitigation strategy obtained from our study.

% Multiple countries are now experiencing the second and third wave of the COVID-19 epidemic. Thus, gaining an understanding of how this outbreak spread in the form of waves and possible interventions in a long term basis is urgent.

%%%%%%% Why mathematical modelling, Recent works   %%%%%%%%%%%%%%

%\citet{yan2020impact} studied a mathematical model to explore how media reports impacts individual behaviour  for mitigating  COVID-19 epidemic. \citet{choi2020optimal} studied a game-theoretic model of COVID-19 transmission dynamics  with  vaccination and social distancing as a control strategy. 
%There are two seminal works on COVID-19 epidemic, that consider the free virus transmission in respective mathematical models \cite{saldana2020modeling,mwalili2020seir}, however, at that time, no clinical data was not found for airborne transmission, and thus transmission dynamics of COVID-19 was not fully explored. Also, several control strategies such as vaccination, anti-covid drug, was not considered for combat with the disease.
%%%%%%%%%%%%%   What we study                               %%%%%%%%%%%%%%%%

%In theory, transitions can happen in continuous time, but for practical implementation we discretize time into days.

%%%%%%%%%%%%%%   Orientation  %%%%%%%%%%%%%%%%%%%%%%%%%%
\section{Model description}\label{Sec:Model}
To design the basic mathematical model for COVID-19 transmission dynamics, some general epidemiological factors are considered including the airborne transmission of the free virus. Several studies of novel corona virus suggest that an infected person can be asymptomatic (infectious but not symptomatic) or symptomatic with small, moderate, severe, or critical symptoms \cite{who2020coronavirus, xu2020clinical, liu2020viral, wu2020characteristics}. For simplicity of model formulation, we break the infected individuals in two separated classes, notified and un-notified. Here, the un-notified class contains untested asymptomatic, asymptomatic who are tested negative, untested symptomatic, and symptomatic who are tested negative, whereas notified class contains those symptomatic and asymptomatic persons who are tested COVID-19 positive. Population among both the notified and un-notified class can be hospitalized, when they have a critical health situation. Thus the entire population can be stratified into six main compartments based on the status of the disease: susceptible ($S(t)$), exposed ($E(t)$), un-notified ($I_u(t)$), notified ($I_n(t)$), hospitalized  ($I_h(t)$), and recovered ($R(t)$), at time $t$. Consequently, the total population size is given by $N(t)=S(t)+E(t)+I_u(t)+I_n(t)+I_h(t)+R(t)$. In addition, $V(t)$ describes the amount of free virus in the environment.\\

Now, we consider individuals from each human class have a per capita natural mortality rate $\mu$. The net influx rate of susceptible population per unit time is $\Pi$. This parameter includes new births, immigration, and emigration of susceptible persons. Also, the susceptible population decrease after infection, acquired through the direct contact of notified or un-notified infected individuals. Let $\beta_1$ be the transmission rate for direct contact with the modification factor $\nu$ for notified infected individuals. Then $\nu\in(0,1)$, since the notified persons are advised to take preventive measures like face mask use, social distancing more seriously. In addition, the susceptible population is infected through the contact of the free virus in air \cite{greenhalgh2021ten,addleman2021mitigating}. Let $\beta_2$ be the transmission rate for free virus contact. The interaction of susceptible individuals with infected individuals (un-notified or notified) follows \textit{standard mixing incidence} and with free virus follows \textit{mass action incidence}. For basic details and difference between \textit{standard mixing incidence} and \textit{mass action incidence}, interested readers are referred to Martcheva's book (\cite{martcheva2015introduction}; Chap. 3). Then the differential equation that describes the rate of change of susceptible individuals at time $t$, is given by 

$$\displaystyle{\od{S}{t}} =\Pi-\frac{\beta_1 S( I_u+ \nu I_n)}{N} - \beta_2 SV -\mu S + \theta R.$$

Here, $\theta$ be the rate at which recovered individuals eventually lose the temporal immunity from the infection and become susceptible \cite{sariol2020lessons, foy2021comparing}. 

\begin{figure}[h]
	\includegraphics[width = 1.0\textwidth]{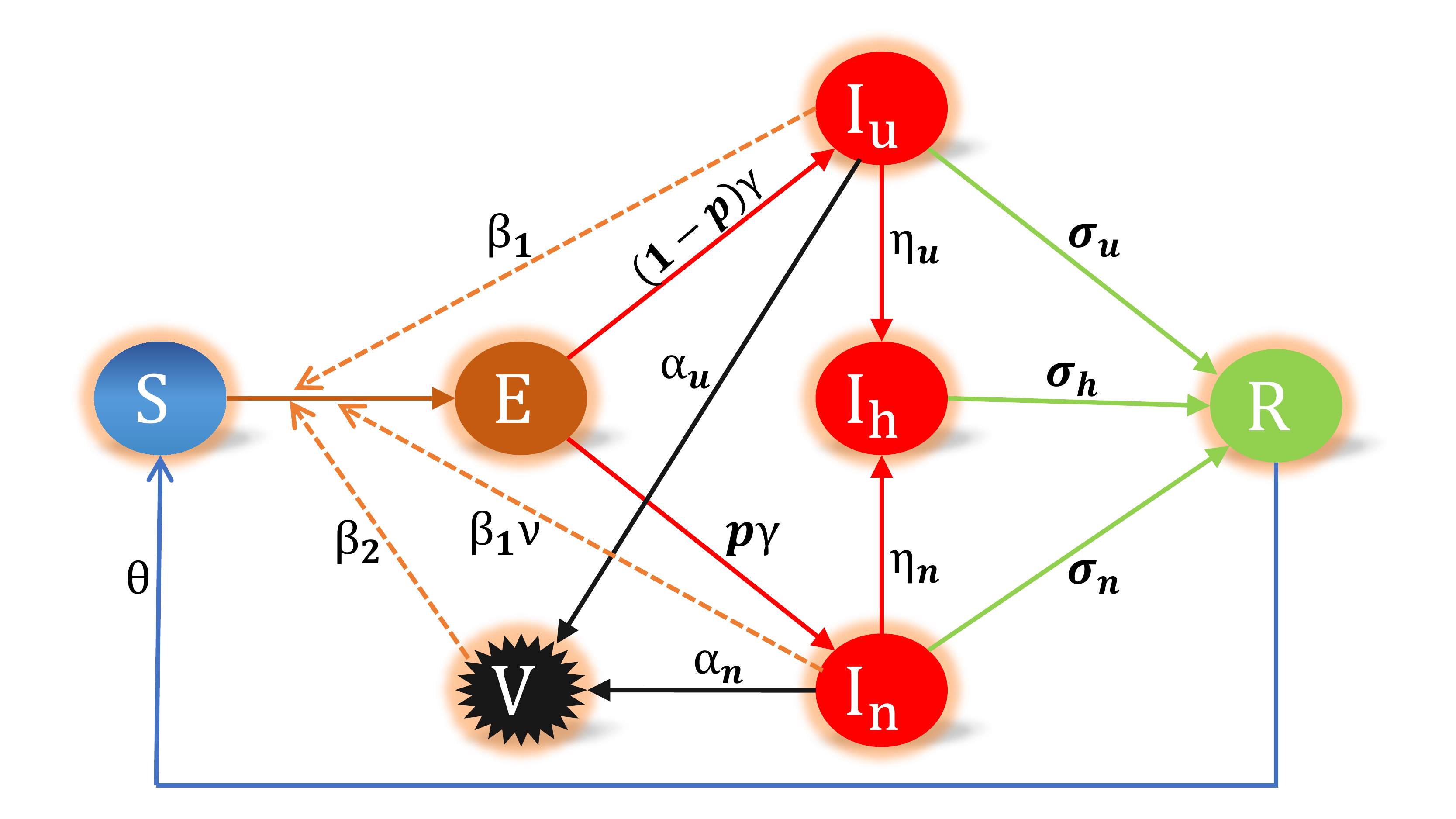}
	\caption{Schematic diagram of the proposed model. Solid arrows represent the transmission rate from one compartment to other, whereas dashed arrows represent interaction between compartments.}
	\label{flowchart}
\end{figure}

Alongside, the exposed population ($E(t)$) who are infected individuals, however still not infectious for the community. By making successful contact with infectives, susceptible individuals become exposed. Also, these people decreases at a rate $\gamma$ and become notified or un-notified. The probability that exposed individuals progress to the notified infectious class is $p$ and that of un-notified class is $1-p$. This assumption leads to the following rate of change in $E$:

$$\displaystyle{\od{E}{t}} = \frac{\beta_1 S( I_u+ \nu I_n)}{N} + \beta_2 SV - (\gamma + \mu) E. $$

The un-notified individuals ($I_u$), who are transferred from exposed class at a rate $(1-p)\gamma$ are the population who are un-notified from COVID-19 disease. Also, the un-notified individuals progress to hospitalized class at a rate $\eta_u$ and recovered at a rate $\sigma_u$. The un-notified individuals are assumed to have no disease-induced mortality rate as most of them will not develop severe symptoms. However, if some of the un-notified persons develop sudden severe symptoms, they will be transferred to the hospital and die as a hospitalized patients. Thus, the equation governing the rate of change of un-notified persons is 

$$\displaystyle{\od{I_u}{t}} =(1-p)\gamma E- (\eta_u+\sigma_u+\mu)I_u.$$

The notified individuals ($I_n$), who are transferred from exposed class at a rate $p\gamma$ are the population who are notified as COVID-19 positive people. Also, the notified individuals reduced by progressing to hospitalized class at a rate $\eta_n$, by recovery at a rate $\sigma_n$, and COVID-19 induced mortality rate $\delta_n$. Therefore we have the following equation 

$$\displaystyle{\od{I_n}{t}} = p \gamma E-(\eta_n + \sigma_n + \mu + \delta_n) I_n. $$

The hospitalized individuals ($I_h$) increase from un-notified or notified people at rates $\eta_n$ or $\eta_u$, respectively. Also, this class of people recover at a rate $\sigma_h$ and decrease through COVID-19 induced mortality rate $\delta_h$. Thus, we have the following rate of change of $I_h$:

$$\displaystyle{\od{I_h}{t}} = \eta_u I_u+ \eta_n I_n - (\sigma_h + \mu + \delta_h)I_h.$$

The recovered individuals ($R$) increase by recovering from each of the infected classes (un-notified, notified and hospitalized) at rates $\sigma_u$, $\sigma_n$, and $\sigma_h$  respectively. As mentioned earlier, $\theta$ is the rate at which recovered population loss immunity against COVID-19 and becomes susceptible again. Thus, following differential equation leads to the rate of change of $R$:

$$\displaystyle{\od{R}{t}} =\sigma_u I_u+\sigma_n I_n + \sigma_h I_h - (\theta + \mu)R.$$

The differential equation describing the rate of change of free virus, increased from the shedding rate of un-notified and notified patients at rates $\alpha_u$ and $\alpha_n$, respectively with a natural clearance rate $\mu_c$. Thus 

$$\displaystyle{\od{V}{t}} = \alpha_u I_u+ \alpha_n I_n - \mu_c V.$$

Fig. \ref{flowchart} describes the flow patterns of individuals between compartments over time.

Assembling all the above differential  equations, we have the following system 
\begin{eqnarray}
	\begin{cases}
		\begin{array}{lll}
			\displaystyle{\od{S}{t}} &=& \Pi-\frac{\beta_1 S( I_u+ \nu I_n)}{N} - \beta_2 SV -\mu S + \theta R, \\
			\displaystyle{\od{E}{t}} &=& \frac{\beta_1 S( I_u+ \nu I_n)}{N} + \beta_2 SV - (\gamma + \mu) E, \\
			\displaystyle{\od{I_u}{t}} &=& (1-p) \gamma E - (\eta_u + \sigma_u +\mu )I_u, \\
			\displaystyle{\od{I_n}{t}} &=&  p \gamma E-(\eta_n + \sigma_n + \mu + \delta_n) I_n, \\
			\displaystyle{\od{I_h}{t}} &=& \eta_u I_u+ \eta_n I_n - (\sigma_h + \mu + \delta_h)I_h,  \\
			\displaystyle{\od{R}{t}} &=& \sigma_u I_u+\sigma_n I_n + \sigma_h I_h - (\theta + \mu)R, \\
			\displaystyle{\od{V}{t}} &=& \alpha_u I_u+ \alpha_n I_n - \mu_c V.
		\end{array}
		\label{EQ:TheModel}
	\end{cases}
\end{eqnarray}

All the parameters involved in the above system are positive. We present the parameters with their biological interpretation, dimensions and realistic values in Table \ref{tab:mod1}.\\

\begin{table}[htbp!]
	\begin{center}
		\caption{Description of model parameters used in the model (\ref{EQ:TheModel}).}
		\label{tab:mod1}
		\begin{tabular}{cp{2.7cm} p{5cm}p{2cm}p{3cm}}
			\hline\hline
			\textbf{Parameter} & \textbf{Dimension} & \textbf{Interpretation} &  \textbf{Value(s)} & \textbf{Reference} \\ \hline
			$\Pi$ & person day$^{-1}$ & Recruitment rate of susceptibles & $\mu \times N_{init}$ & -- \\
			$\beta_1 $& day$^{-1}$ & Transmission rate for direct contact & (0,1) & To be estimated \\
			$\beta_2 $ & copies$^{-1}$ litre day$^{-1}$& Transmission rate for free virus contact & (0,1) & To be estimated \\
			$\nu$ & unitless & Modification factor for notified infectives & 0.1852 & \cite{ghosh2021modeling} \\
			$\theta$ & day$^{-1}$ & Rate at which recovered individuals lose immunity & 1/365 & \cite{foy2021comparing} \\
			$\mu$ & day$^{-1}$ & Natural mortality rate & 0.3891 $\times {10}^{-4}$ & \cite{Worldometer2021demo}  \\
			$\gamma$& day$^{-1}$ & Rate at which the exposed individuals are infected & 0.2  & \cite{li2020early,linton2020incubation} \\
			$\eta_u$ & day$^{-1}$ & Rate at which un-notified patients become hospitalized & (0,1) & To be estimated \\
			$\eta_n$ & day$^{-1}$ & Rate at which notified patients become hospitalized & (0,1) & To be estimated \\
			$p$ & unitless & Proportion of notified infectives & 0.2  & \cite{wu2020characteristics,yang2020epidemiological}\\%infectious
			$\sigma_u$ & day$^{-1}$ & Recovery rate of un-notified patients & 0.17 & \cite{tindale2020transmission,woelfel2020clinical} \\
			$\sigma_n$ & day$^{-1}$  & Recovery rate of notified patients & 0.072 &\cite{lopez2020end} \\
			$\sigma_h$ & day$^{-1}$ & Recovery rate of hospitalized patients & (0,1) & To be estimated \\
			$\delta_n$ & day$^{-1}$ & Disease induced mortality rate of notified patients & 0.0017 & \cite{ghosh2021modeling} \\
			$\delta_h$ & day$^{-1}$ & Disease induced mortality rate of hospitalized patients & (0,1) & To be estimated \\
			$\alpha_u$ & copies litre$^{-1}$ person$^{-1}$ day$^{-1}$ & Virus shedding rate of un-notified patients & (0,10) & To be estimated \\
			$\alpha_n$ & copies litre$^{-1}$ person$^{-1}$ day$^{-1}$ & Virus shedding rate of notified patients & (0,10) & To be estimated \\
			$\mu_c$ & day$^{-1}$ & Natural clearance rate of free virus & 1 & \cite{saldana2020modeling} \\
			\hline\hline
\end{tabular}\end{center}\end{table}

\section{Mathematical analysis}\label{sec:math}
\subsection{Basic properties}\label{subsec:basic}
Here, we explore basic dynamical properties of the model (\ref{EQ:TheModel}). Let us first consider the initial conditions of the model as 
\begin{equation}
	S(0)>0,~E(0),~I_u(0),~I_n(0),~I_h(0),~R(0),~V(0)\geq 0.\label{model_IC}
\end{equation}
The dynamical system (\ref{EQ:TheModel}) is well posed as shown in the following theorem.
\begin{theorem}
	Consider the model (\ref{EQ:TheModel}) with initial conditions (\ref{model_IC}). The nonnegative orthant $\mathbb{R}_+^7$ is invariant under the flow of (\ref{EQ:TheModel}), including $S$ remaining positive with the advancement of time. Moreover, the solutions of the system (\ref{EQ:TheModel}) are bounded in the region
	
	$$\Omega=\left\{(S,~E,~I_u,~I_n,~I_h,~R,~V)|~~0<S+E+I_u+I_n+I_h+R\leq \frac{\Pi}{\mu},~~ 0\leq V<\frac{\Pi(\alpha_u+\alpha_n) }{\mu\mu_c}\right\}.$$
\end{theorem}
\begin{proof}
	Consider the initial conditions (\ref{model_IC}). Suppose that at time $t=t_1$, $S(t_1)=0$. Then from (\ref{EQ:TheModel}) at $t=t_1$, $\od{S}{t}=\Pi+\theta R>0$ which implies that $\od{S}{t}>0$ when $S$ is positive and small. Thus, there is no time $t_1$ such that $S(t_1)=0$. Therefore, $S$ stays positive for $t>0$ with the initial condition $S(0)>0$. Now, for the other components 
	\begin{align*}
		\evalat{\od{E}{t}}{E=0}&=\frac{\beta_1 S( I_u+ \nu I_n)}{N} + \beta_2 SV\geq0, \quad
		&\evalat{\od{I_u}{t}}{I_u=0}&= (1-p) \gamma E\geq0, \\
		\evalat{\od{I_n}{t}}{I_n=0} &=  p \gamma E\geq0, \quad
		&\evalat{\od{I_h}{t}}{I_h=0} &= \eta_u I_u+ \eta_n I_n,\geq0, \\
		\evalat{\od{R}{t}}{R=0} &= \sigma_u I_u+\sigma_n I_n + \sigma_h I_h,\geq0, \quad
		& \evalat{\od{V}{t}}{V=0}&=\alpha_u I_u+ \alpha_n I_n\geq0.
	\end{align*}
	Thus, all the other components are nonnegative and the orthant $\mathbb{R}_+^7$ is invariant under (\ref{EQ:TheModel}).\\
	
	To show the boundedness of solutions of the system (\ref{EQ:TheModel}), we add all the equations except the last one of (\ref{EQ:TheModel}) to get
	\begin{align*}
		\od{N}{t}&=\od{}{t}(S+E+I_u+I_n+I_h+R)\\
		&=\Pi-\mu N-\delta_nI_n-\delta_hI_h\\
		&\leq  \Pi-\mu N.
	\end{align*}
	Using the results of differential inequality, we have
	
	$$\limsup_{t\to\infty}N(t)\leq\frac{\Pi}{\mu}.$$
	Since the total population is bounded, so the individual components are also bounded. \\
	Also, from the last equation of (\ref{EQ:TheModel}), the rate of change of free virus is given by
	\begin{align*}
		\od{V}{t} &= \alpha_u I_u+ \alpha_n I_n - \mu_c V\\
		&< (\alpha_u+ \alpha_n)\frac{\Pi}{\mu}- \mu_c V
	\end{align*}
	Again, using the results of differential inequality, we have
	$$\limsup_{t\to\infty}V(t)<\frac{\Pi(\alpha_u+ \alpha_n)}{\mu\mu_c}.$$
	This completes the proof.
\end{proof}

\subsection{Disease-free equilibrium and basic reproduction number}
As of many epidemiological models, our model (\ref{EQ:TheModel}) also exhibit the 
disease-free equilibrium (DFE). At this equilibrium, population persists in the absence of disease. Mathematically, DFE is obtained by assuming all the infected compartments to be zero, i.e., $E, I_u, I_n, I_h$, and $V$ to be zero. Then from the model, $S=\frac{\Pi}{\mu}$ and $R=0$. Thus, the DFE is given by $\mathscr{E}_s=(\frac{\Pi}{\mu},0,0,0,0,0,0)$.

The basic reproduction number $\mathscr{R}_0$ is interpreted as the average number of secondary cases produced by one infected individual introduced into a population of susceptible individuals. This dimensionless number is a measure of the potential for disease outbreak. Now to obtain $\mathscr{R}_0$, we use next-generation operator method \cite{diekmann1990definition,van2002reproduction}.
%basic reproduction number is interpret as the expected number of secondary infections produced by an index case in a completely susceptible population

First we assemble the compartments which are infected (i.e., $E$, $I_u$, $I_n$, $I_h$, and $V$) and  decomposing the right hand side of the system (\ref{EQ:TheModel}) as $\mathscr{F}-\mathscr{W}$, where $\mathscr{F}$ is the transmission part, expressing the production of new infection, and the transition part is $\mathscr{W}$, which
describe the rate of transfer of individuals from one compartment to another. Then we have for the model (\ref{EQ:TheModel})
$$\mathscr{F}=\begin{pmatrix}
	\frac{\beta_1 S( I_u+ \nu I_n)}{N}+ \beta_2 SV\\
	0\\
	0\\
	0\\
	0
\end{pmatrix}\quad \text{and}\quad
\mathscr{W}=\begin{pmatrix}
	(\gamma + \mu) E\\
	-(1-p) \gamma E+\phi_uI_u\\
	-p\gamma E+\phi_n I_n \\
	-\eta_u I_u- \eta_n I_n +\phi_hI_h \\
	-\alpha_u I_u-\alpha_n I_n+ \mu_c V 
\end{pmatrix},$$
with $\phi_u=\eta_u+\sigma_u+\mu$,  $\phi_n=\eta_n+\sigma_n+\mu+\delta_n$, and $\phi_h=\sigma_h+\mu+\delta_h$.\\
Consider $X=(E, I_u, I_n, I_h, V)$.
Then the derivatives of $\mathscr{F}$ and $\mathscr{W}$ at the DFE $\mathscr{E}_s$ are given by
$$F=\pd{\mathscr{F}}{X}=\begin{pmatrix}
	0 &  \beta_1 &  \beta_1\nu & 0 &\frac{\beta_2\Pi}{\mu}\\
	0 &   0      &   0         & 0  & 0\\
	0 &   0      &   0         & 0  & 0\\
	0 &   0      &   0         & 0  & 0\\
	0 &   0      &   0         & 0  & 0
\end{pmatrix}\quad \text{and}\quad
W=\pd{\mathscr{W}}{X}=\begin{pmatrix}
	(\gamma + \mu) &   0         &   0         & 0    & 0\\
	-(1-p) \gamma  &  \phi_u     &   0         & 0    & 0\\
	-p\gamma       &  0          &  \phi_n     & 0    & 0 \\
	0              &  -\eta_u    &  -\eta_n    &\phi_h& 0\\
	0              & -\alpha_u   &  -\alpha_n  & 0    &\mu_c 
\end{pmatrix}.$$
Now, we calculate the next-generation matrix $FW^{-1}$, whose $(i,j)$-th entry describe the expected number of new infections in compartment $i$ produced by the infected individual originally introduced into compartment $j$. According to Diekmann et al. \cite{diekmann1990definition}, the basic reproduction number is given by $\mathscr{R}_0=\rho(FW^{-1})$, where $\rho$ is the spectral radius of the next-generation matrix $FW^{-1}$. A simple calculation leads to the following expression for $\mathscr{R}_0$:
\begin{equation}
	\mathscr{R}_0=\frac{(1-p)\gamma\beta_1}{\phi_u( \gamma+\mu)}+\frac{p\gamma\nu\beta_1}{\phi_n( \gamma+\mu)}+\frac{(1-p)\gamma\alpha_u\beta_2(\Pi/\mu)}{\mu_c \phi_u( \gamma+\mu)}+\frac{p\gamma\alpha_n\beta_2(\Pi/\mu)}{\mu_c\phi_n( \gamma+\mu)}.\label{EQ:basic_reproduction}
\end{equation}
Epidemiologically, the first and second term of  $\mathscr{R}_0$ can be interpreted as the number of secondary infections that one un-notified and notified individual will produce in a completely susceptible population during its infectious period, respectively. Here, the infection occurs due to direct contact. Also, the last two term  represent the same but for free virus contact. Thus, both transmission pathways have impacts on the dynamics of COVID-19 spread in the community.

\subsection{Stability of disease-free equilibrium}\label{stability_DFE}
Following theorem gives the local stability properties of the DFE $\mathscr{E}_s=(\frac{\Pi}{\mu},0,0,0,0,0,0)$. 
\begin{theorem}
	The disease-free equilibrium (DFE) $\mathscr{E}_s=(\frac{\Pi}{\mu},0,0,0,0,0,0)$ of the system (\ref{EQ:TheModel}) is locally asymptotically stable if $\mathscr{R}_0<1$ and unstable if $\mathscr{R}_0>1$.
\end{theorem}
\begin{proof}
	The Jacobian matrix at the DFE $\mathscr{E}_s$ is given by
	\begin{equation}
		J_{\mathscr{E}_s}= \left[ \begin {array}{ccccccc} 
		-\mu &    0        &-\beta_1    &-\nu\beta_1 & 0       & \theta     & -\frac {\Pi\beta_2}{\mu}\\
		0   & -(\gamma+\mu)& \beta_1    & \nu\beta_1 & 0       & 0          & \frac{\Pi\beta_2}{\mu}\\
		0    & (1-p)\gamma &-\phi_u     &  0         & 0       & 0          & 0\\
		0    & p\gamma     & 0          & -\phi_n    & 0       & 0          & 0\\
		0    & 0           &\eta_u      & \eta_n     &-\phi_h  & 0          & 0\\
		0    & 0           &\sigma_u    &\sigma_n    &\sigma_h &-(\theta+\mu) & 0 \\
		0    &0            &\alpha_u    &\alpha_n    &  0      &0           &-\mu_c
		\end {array} \right],\label{Jacobian_DFE}
	\end{equation}
	with $\phi_u=\eta_u+\sigma_u+\mu$, $\phi_n=\eta_n+\sigma_n+\mu+\delta_n$, and $\phi_h=\eta_h+\mu+\delta_h$.\\
	The above Jacobian matrix possess three obvious eigenvalues $-\mu$, $-(\theta+\mu)$ and $-\phi_h=-(\sigma_h+\mu+\delta_h)$, which are all negative. Also, the remaining eigenvalues ($\lambda$) are the roots of the equation $\det(M-\lambda I)=0$, where $I$ is the $4\times 4$ identity matrix and
	$$M= \left[ \begin {array}{ccccccc} 
	-(\gamma+\mu)  & \beta_1               & \nu\beta_1                      & \frac {\Pi\beta_2}{\mu}                     \\
	(1-p)\gamma   &-\phi_u &  0                              & 0                       \\
	p\gamma      & 0                     & -\phi_n & 0                       \\
	0             &\alpha_u                 & \alpha_n                          &-\mu_c  \\
	\end {array} \right].$$
	Now, $\det(M-\lambda I)=0$ gives 
	\begin{align*}
		&(\lambda+\mu_c)(\lambda+\phi_n)(\lambda+\phi_u)(\lambda+\gamma+\mu)-(1-p)\gamma\beta_1(\lambda+\mu_c)(\lambda+\phi_n)\\
		&-p\gamma\nu\beta_1(\lambda+\mu_c)(\lambda+\phi_u)
		-(1-p)\gamma\alpha_u\frac{\Pi\beta_2}{\mu}(\lambda+\phi_n)-\frac {\Pi\beta_2}{\mu}p\gamma\alpha_n(\lambda+\phi_u)=0. \end{align*}
	The above equation can be written as $H(\lambda)=1$, where
	\begin{align}
		\begin{split}
			H(\lambda)=&\frac{(1-p)\gamma\beta_1}{(\lambda+\phi_u)(\lambda+\gamma+\mu)}+\frac{p\gamma\nu\beta_1}{(\lambda+\phi_n)(\lambda+\gamma+\mu)}+\frac{(1-p)\gamma\alpha_u\beta_2(\Pi/\mu)}{(\lambda+\mu_c)(\lambda+\phi_u)(\lambda+\gamma+\mu)}\\
			&+\frac{p\gamma\alpha_n\beta_2(\Pi/\mu)}{(\lambda+\mu_c)(\lambda+\phi_n)(\lambda+\gamma+\mu)}.\label{h_lambda}
		\end{split}
	\end{align}
	We also rewrite $H(\lambda)$ as $$H(\lambda)=H_1(\lambda)+H_2(\lambda)+H_3(\lambda)+H_4(\lambda),$$
	where $H_j(\lambda)$, $j=1,2,3,4$, are the respective term of (\ref{h_lambda}).
	Now, if $Re(\lambda)\geq0$ with $\lambda=x+iy$, then
	$$|H_1(\lambda)|=\frac{(1-p)\gamma\beta_1}{|\lambda+\phi_u||\lambda+\gamma+\mu|}\leq H_1(x)\leq H_1(0).$$
	Analogously, $$|H_j(\lambda)|\leq H_j(x)\leq H_j(0),\quad \text{for}~j=2,3,4.$$
	Then 
	\begin{align*}
		|H(\lambda)|\leq& |H_1(\lambda)|+|H_2(\lambda)|+|H_3(\lambda)|+|H_4(\lambda)|\\
		\leq &H_1(0)+H_2(0)+H_3(0)+H_4(0)\\
		= &H(0).
	\end{align*}
	From (\ref{EQ:basic_reproduction}) and (\ref{h_lambda}), we have $H(0)=\mathscr{R}_0$. Then $\mathscr{R}_0<1$ $\implies$ $|H(\lambda)|<1$, which implies there does not exist any solutions of $H(\lambda)= 1$ with  $Re(\lambda)\geq0$.\\
	
	Therefore if $\mathscr{R}_0<1$, all the eigenvalues of $H(\lambda)= 1$ have negative real parts and hence the DFE $\mathscr{E}_s$ of the system (\ref{EQ:TheModel}) is locally asymptotically stable.\\
	Now for the case of $\mathscr{R}_0>1$, i.e., $H(0)>1$ 
	$$\lim_{\lambda\to\infty}H(\lambda)=0.$$
	Then there exists $\lambda^*>0$ such that $H(\lambda^*)=1$, which confirm the existence of positive eigenvalue of the Jacobian matrix $J_{\mathscr{E}_s}$.
	Thus the DFE $\mathscr{E}_s$ is unstable for $\mathscr{R}_0>1$.
\end{proof}

Moreover, when $\mathscr{R}_0<1$ the DFE $\mathscr{E}_s$ is globally asymptotically stable, which can be assured by the following theorem.
\begin{theorem}
	The disease-free equilibrium (DFE) $\mathscr{E}_s=(\frac{\Pi}{\mu},0,0,0,0,0,0)$ of the system (\ref{EQ:TheModel}) is globally asymptotically stable if $\mathscr{R}_0<1$.
\end{theorem}
\begin{proof}
	First we rewrite the model (\ref{EQ:TheModel}) into the form 
	\begin{eqnarray}
		\begin{cases}
			\begin{array}{lll}
				\displaystyle{\od{\mathbb{X}}{t}} &=& P(\mathbb{X},\mathbb{Z}), \\
				\displaystyle{\od{\mathbb{Z}}{t}} &=& Q(\mathbb{X},\mathbb{Z})
			\end{array}
			\label{EQ:global_DFE}
		\end{cases}
	\end{eqnarray}
	where $\mathbb{X}^T=(S,R)\in\mathbb{R}^2_+$ with $S>0$, describe the uninfected compartments and $\mathbb{Z}^T=(E,I_u,I_n,I_h,V)\in\mathbb{R}^5_+$ describe the infected compartments with free virus. Also, $T$ denoting the transpose of the matrix. Note that $Q(\mathbb{X},\mathbf{0})=\mathbf{0}$, where $\mathbf{0}$ is a zero vector.
	
	Following Castillo-Chavez \cite{castillo2002computation}, the DFE $\mathscr{E}_s$ of the system (\ref{EQ:global_DFE}) is globally asymptotically stable if it is locally asymptotically stable and satisfy following two conditions:
	\begin{enumerate}[(H1)]
		\item For the subsystem $\od{\mathbb{X}}{t}=P(\mathbb{X},\mathbf{0})$, the equilibrium $\mathbb{X}^*$ is globally asymptotically stable,
		\item $Q(\mathbb{X},\mathbb{Z})=B\mathbb{Z}-\widetilde{Q}(\mathbb{X},\mathbb{Z})$, $\widetilde{Q}(\mathbb{X},\mathbb{Z})\geq 0$ for $(\mathbb{X},\mathbb{Z})\in\Omega$, 
	\end{enumerate}
	where $B=\evalat{\frac{\partial Q}{\partial \mathbb{Z}}}{(\mathbb{X},\mathbf{Z})=(\mathbb{X}^*,\mathbf{0})}$ is a Metzler matrix (a matrix whose off diagonal elements are nonnegative) and $\Omega$ is the positive invariant set for the model (\ref{EQ:TheModel}) as described in Subsection \ref{subsec:basic}.\\
	Now, for the model (\ref{EQ:TheModel}),
	$$P(\mathbb{X},\mathbf{0})=\begin{pmatrix}
		\Pi-\mu S\\0
	\end{pmatrix}.$$
	Since all the infected compartments are zero, so there is no infection, and thus, no recovery. That's why, we consider $R=0$. Clearly, $\mathbb{X}^*=(\frac{\Pi}{\mu},0)$ is a globally asymptotically equilibrium for the system  $\od{\mathbb{X}}{t}=P(\mathbb{X},\mathbf{0})$. So the condition (H1) satisfied.
	Now, 
	$$B=\begin{pmatrix}
		-(\gamma+\mu)  & \beta_1  &  \beta_1\nu  &  0 & \frac{\beta_2\Pi}{\mu}\\
		(1-p)\gamma & -\phi_u &  0 & 0  & 0\\
		p\gamma     &  0 & -\phi_n & 0  & 0\\
		0 & \eta_u &  \eta_n & -\phi_h  & 0\\
		0 & \alpha_u &  \alpha_n & 0  & -\mu_c\\
	\end{pmatrix}\quad \text{and}$$
	$$\widetilde{Q}(\mathbb{X},\mathbb{Z})=\begin{pmatrix}
		\beta_1(1-\frac{S}{N})I_u+\beta_1\nu (1-\frac{S}{N})I_n+\beta_2(\frac{\Pi}{\mu}-S)V\\
		0\\
		0\\
		0\\
		0
	\end{pmatrix},$$
	with $\phi_u=\eta_u+\sigma_u+\mu$, $\phi_n=\eta_n+\sigma_n+\mu+\delta_n$, and $\phi_h=\eta_h+\mu+\delta_h$.\\
	The matrix $B$ is a Metzler matrix and $\widetilde{Q}(\mathbb{X},\mathbb{Z})\geq0$ whenever the state variables are inside the invariant set $\Omega$. Thus (H2) also satisfied. This completes the proof.
\end{proof}

\subsection{Existence of endemic equilibrium}
Here, we discuss the existence of endemic equilibrium of the model (\ref{EQ:TheModel}). The endemic equilibrium $\mathscr{E}^*=(S^*,E^*,I_u^*,I_n^*,I_h^*,R^*,V^*)$ is given by
$I_u^*=A_1E^*$, $I_n=A_2E^*$, $I_h^*=A_3E^*$, $R=A_4E^*$, $V^*=A_5E^*$,  $S^*=\frac{\Pi+(\theta A_4-(\gamma+\mu))E^*}{\mu}=\frac{(\gamma+\mu)(\Pi-A_7E^*)}{A_6+\beta_2A_5(\Pi-A_7E^*)}$, 
where all the $A_i$, $i=1,2,\cdots,7$ are positive and given by
\begin{align*}
	A_1&=\frac{(1-p)\gamma}{\eta_u+\sigma_u+\mu},\; & A_2&=\frac{p\gamma}{\eta_n+\sigma_n+\mu+\delta_n},\; &
	A_3&=\frac{\eta_u A_1+\eta_n A_2}{\sigma_h+\mu+\delta_h},\\
	A_4&=\frac{\sigma_u A_1+\sigma_n A_2+\sigma_h A_3}{\theta+\mu},\;& 
	A_5&=\frac{\alpha_u A_1+\alpha_n A_2}{\mu_c},\; &
	A_6&=\mu\beta_1(A_1+\nu A_2),\;\\ A_7&=\delta_n A_2+\delta_h A_3,
\end{align*}
and $E^*$ is a positive root of the quadratic equation 
\begin{equation}
	\rho_0E^2+\rho_1E+\rho_2=0\label{root_endemic}
\end{equation}
with
\begin{align*} \rho_0&=-\beta_2 A_5 A_7(\theta A_4-(\gamma+\mu)),\\
	\rho_1&=(A_6+\beta_2 A_5 \Pi)(\theta A_4-(\gamma+\mu))-\beta_2 A_5 A_7\Pi+\mu(\gamma+\mu )A_7\\
	&=\mu(\gamma+\mu)\mathscr{R}_0(\theta A_4-(\gamma+\mu))+\mu(\gamma+\mu)(1-\mathscr{R}_0)A_7+A_6A_7,\\
	\rho_2&=(A_6+\beta_2 A_5 \Pi)\Pi-\mu(\gamma+\mu)\Pi
	=\mu(\gamma+\mu)\Pi(\mathscr{R}_0-1).
\end{align*}
For the existence of real roots of (\ref{root_endemic}), we must have $\rho_1^2-4\rho_0\rho_2\geq0$. Without loss of generality, assume that $E_1^*>E_2^*$, when equation (\ref{root_endemic}) have two real roots.
From the expression of $S^*$, $S^*=\frac{\Pi}{\mu}+\frac{(\theta A_4-(\gamma+\mu))}{\mu}E^*$.  Since the upper bound of total population is $\Pi/\mu$, so we must have $\theta A_4-(\gamma+\mu)< 0$. Then $\rho_0$ is always positive. Note that, for  $\theta A_4-(\gamma+\mu)= 0$, $S^*=\frac{\Pi}{\mu}$ and thus the endemic equilibrium becomes DFE. Also, $\rho_2<(>)0\iff\mathscr{R}_0<(>)1$. At the equilibrium density, $N^*=\frac{\Pi-A_7E^*}{\mu}$. Thus, $N^*>0\iff \Pi-A_7E^*>0$.
%Now, depending on the number of positive solution of (\ref{root_endemic}), we consider the following cases.
Keeping all the conditions in mind, in the next theorem, we state about the existence and possible number of endemic equilibrium points.
\begin{theorem}
	When $\mathscr{R}_0<1$, the model (\ref{EQ:TheModel}) has unique endemic equilibrium if $\Pi-A_7E^*>0$. When $\mathscr{R}_0>1$, the model (\ref{EQ:TheModel}) has
	\begin{enumerate}
		\item two endemic equilibrium if $\rho_1<0$,  $\rho_1^2-4\rho_0\rho_2>0$, and $\Pi-A_7E_i^*>0$, $i=1,2$.
		\item unique endemic equilibrium if any of the following cases holds
		\begin{enumerate}
			\item $\rho_1<0$,  $\rho_1^2-4\rho_0\rho_2>0$, and $E_1^*>\frac{\Pi}{A_7}>E_2^*$,
			\item $\rho_1<0$,  $\rho_1^2-4\rho_0\rho_2=0$, $\Pi-A_7E^*>0$,
		\end{enumerate} 
		\item no endemic equilibrium otherwise.
	\end{enumerate} 
\end{theorem}

%\subsection{Local stability of DFE $\mathscr{E}_s$ when $\mathscr{R}_0=1$}
\subsection{Analysis of the center manifold near DFE $\mathscr{E}_s$ when $\mathscr{R}_0=1$}
In this subsection, we discuss the nature of DFE $\mathscr{E}_s$ when $\mathscr{R}_0=1$. Recall the Jacobian matrix $J_{\mathscr{E}_s}$ from (\ref{Jacobian_DFE}). Notice that, when $\mathscr{R}_0=1$, $J_{\mathscr{E}_s}$ possess a zero eigenvalue. So the equilibrium $\mathscr{E}_s$ becomes non-hyperbolic. Then center manifold theory is the best approach to study the behaviour of this equilibrium. In this regard, we follow Theorem 4.1 described in Castillo-Chavez and Song \cite{castillo2004dynamical}.

To proceed, first we write the mathematical model (\ref{EQ:TheModel}) into the vector form $\od{\mathbf{x}}{t}=f(\mathbf{x})$, with $\mathbf{x}=(x_1,~ x_2,~ x_3,~ x_4,~ x_5,~ x_6,~ x_7)^T=(S,~E,~I_u,~I_n,~I_h,~R,~V)^T$ and $f(\mathbf{x})=(f_1(\mathbf{x}), ~f_2(\mathbf{x}),~ \cdots,~ f_7(\mathbf{x}))^T$. Since $\mathscr{R}_0$ is often inconvenient to use directly as a bifurcation parameter, we introduce $\beta_1$ as a bifurcation parameter. Then $\mathscr{R}_0=1$ gives 
$$\beta_1=\frac{\mu_c (\gamma+\mu)\phi_u \phi_n -\gamma\beta_2(\Pi/\mu)((1-p)\alpha_u\phi_n+p\alpha_n\phi_u)}{\mu_c \gamma((1-p)\phi_n+\nu p\phi_u)}=\beta_1^* ~~(\text{say}).$$ 
Moreover, $\mathscr{R}_0<1$ for $\beta_1<\beta_1^*$ and $\mathscr{R}_0>1$ for $\beta_1>\beta_1^*$. Now, consider the system $\od{\mathbf{x}}{t}=f(\mathbf{x},\beta_1)$ and at $\beta_1=\beta_1^*$ the Jacobian matrix $J_{\mathscr{E}_s}$ have simple zero eigenvalue and all the other eigenvalues have negative real parts. Let $\mathbf{v}=(v_1,~ v_2,~ \cdots~, v_7)$ and $\mathbf{w}=(w_1,~ w_2,~ \cdots~, w_7)^T$ are the respective left and right eigenvectors corresponding to zero eigenvalue. Simple algebraic calculations lead us to the following expressions of $v_i$ and $w_i$, $i=1,2,\cdots, 7$:
\begin{align*}
	v_1&=0=v_5=v_6, \quad v_2=v_2>0,\quad v_3=\frac{1}{\phi_u}\left(\beta_1^*+\frac{\beta_2\alpha_u\Pi}{\mu \mu_c}\right)v_2,\\
	v_4&=\frac{1}{\phi_n}\left(\beta_1^*\nu+\frac{\beta_2\alpha_n\Pi}{\mu \mu_c}\right)v_2,\quad
	v_7=\frac{\beta_2\Pi}{\mu \mu_c}v_2,
\end{align*}
and
\begin{align*}
	w_1&=\frac{1}{\mu}\left[\frac{\theta}{\theta+\mu}\left\{\frac{(1-p)\gamma}{\phi_u}\left(\sigma_u+\frac{\sigma_h\eta_u}{\phi_h}\right)+\frac{p\gamma}{\phi_n}\left(\sigma_n+\frac{\sigma_h\eta_n}{\phi_h}\right)\right\}-(\gamma+\mu)\right]w_2,\\
	w_2&=w_2>0,\quad w_3=\frac{(1-p)\gamma}{\phi_u}w_2,\quad w_4=\frac{p\gamma}{\phi_n}w_2,\quad w_5=\frac{1}{\phi_h}\left(\frac{(1-p)\gamma\eta_u}{\phi_u}+\frac{p\gamma\eta_n}{\phi_n}\right)w_2,\\
	w_6&=\frac{1}{\theta+\mu}\left[\frac{(1-p)\gamma}{\phi_u}\left(\sigma_u+\frac{\sigma_h\eta_u}{\phi_h}\right)+\frac{p\gamma}{\phi_n}\left(\sigma_n+\frac{\sigma_h\eta_n}{\phi_h}\right)\right]w_2,\\
	w_7&=\frac{1}{\mu_c}\left(\frac{(1-p)\gamma\alpha_u}{\phi_u}+\frac{p\gamma\alpha_n}{\phi_n}\right)w_2.
\end{align*}
Clearly, all the $v_i$ and $w_i$ ($i=1,2,\cdots,7$) are positive except $w_1$. From the expression of $w_1$, note that $w_1=\frac{1}{\mu}(\theta A_4-(\gamma+\mu))w_2$. Also, since  $\theta A_4-(\gamma+\mu)<0$, so $w_1<0$.
Now, according to the Theorem 4.1 \cite{castillo2004dynamical}, we calculate the quantities $a$ and $b$, which are defined by
\begin{align}
	a&=\sum_{k,i,j=1}^{7}v_kw_iw_j\frac{\partial^2 f_k}{\partial x_i \partial x_j}\label{center_manifold_a},\\
	b&=\sum_{k,i=1}^{7}v_kw_i\frac{\partial^2 f_k}{\partial x_i \partial \beta_1}\label{center_manifold_b},
\end{align}
where all the partial derivatives are evaluated at the DFE $\mathscr{E}_s$ and $\beta_1=\beta_1^*$. Since $v_1$ is zero and $f_k$, $k=3,4,\cdots 7$ are linear functions of the state variables, so we need to calculate only $\frac{\partial^2 f_2}{\partial x_i \partial x_j}$. Then we have
\begin{align*}
	\frac{\partial^2 f_2}{\partial x_1 \partial x_7}&=\beta_2, \quad \frac{\partial^2 f_2}{\partial x_2 \partial x_3}=-\frac{\beta_1^*\mu}{\Pi}=\frac{\partial^2 f_2}{\partial x_3 \partial x_5}=\frac{\partial^2 f_2}{\partial x_3 \partial x_6},\quad \frac{\partial^2 f_2}{\partial x_3 \partial x_4}=-\frac{\beta_1^*\mu(1+\nu)}{\Pi},\\
	\frac{\partial^2 f_2}{\partial x_3^2}&=-\frac{2\beta_1^*\mu}{\Pi},\quad \frac{\partial^2 f_2}{\partial x_2 \partial x_4}=-\frac{\beta_1^*\mu\nu}{\Pi}=\frac{\partial^2 f_2}{\partial x_4 \partial x_5}=\frac{\partial^2 f_2}{\partial x_4 \partial x_6}, \quad \frac{\partial^2 f_2}{\partial x_4^2}=-\frac{2\beta_1^*\mu\nu}{\Pi},
\end{align*}
and the rest of the derivatives are zero. Thus from (\ref{center_manifold_a}), we have 
\begin{equation*}
	a=v_2\left[\beta_2w_1w_7-\frac{\beta_1^*\mu}{\Pi}\left\{w_3(w_2+2w_3+w_4+w_5+w_6)+\nu w_4(w_2+w_3+2w_4+w_5+w_6)\right\}\right].
\end{equation*}
Again, since $v_1$ is zero and  $\frac{\partial^2 f_k}{\partial x_i \partial \beta_1}=0$ for $k=3,4,\cdots,7$, then from (\ref{center_manifold_b}) we have 
\begin{equation*}
	b=v_2\left[w_3\frac{\partial^2 f_2}{\partial x_3 \partial \beta_1}+w_4\frac{\partial^2 f_2}{\partial x_4 \partial \beta_1}\right]=v_2(w_3+\nu w_4)>0.
\end{equation*}
Since $w_1<0$, so $a<0$.  Hence, for the system (\ref{EQ:TheModel}), $a<0$ and $b>0$. Thus the system exhibit forward bifurcation at $\beta_1=\beta_1^*$, and an endemic equilibrium appears which is locally asymptotically stable.

We numerically verify the existence of forward bifurcation with respect to the basic reproduction number $\mathscr{R}_0$. Since $\mathscr{R}_0$ is not a model parameter and it depends on $\beta_1$, so we vary $\beta_1$. When $\mathscr{R}_0$ crosses the threshold value $1$, the system transits from the stable DFE to the stable endemic equilibrium and the DFE becomes unstable (Fig. \ref{forward_bif}). In the figure, blue-coloured line depicts the stable DFE and red curve depicts the stable endemic equilibrium.

\begin{figure}[H]
	\centering
	\begin{tabular}{cccc}
		\includegraphics[width = 0.7\textwidth]{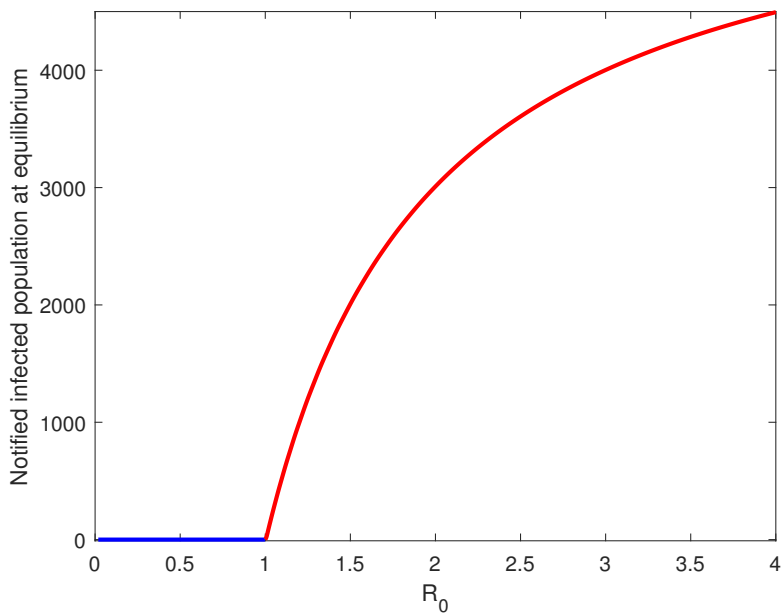}
	\end{tabular}
	\caption{Forward bifurcation of the system (\ref{EQ:TheModel}) with respect to the basic reproduction number $\mathscr{R}_0$. The parameter values are taken as $\Pi=496.68615$, $\beta_2=6.6248\times 10^{-8}$, $\eta_u=0.0223$, $\eta_n = 0.7846$, $\sigma_h= 0.9181$, $\delta_h = 0.0046$, $\alpha_u = 1.3232\times 10^{-6}$, $\alpha_n = 0.0741$ and the other parameter values are same as Table \ref{tab:mod1}.}
	\label{forward_bif}
\end{figure}

We also draw the time series plot by choosing two different values of $\beta_1$. For $\beta_1=0.15$, $\mathscr{R}_0<1$ and all the infected compartments (un-notified, notified and hospitalized) becomes zero and for $\beta_1=0.3$, $\mathscr{R}_0>1$ and infected compartments persists in a stable manner as time evolves (Fig. \ref{time_series}).

\begin{figure}[h]
	%	\centering
	\hspace*{-2cm}
	\begin{tabular}{llll}
		\includegraphics[height = 10cm, width = 20cm]{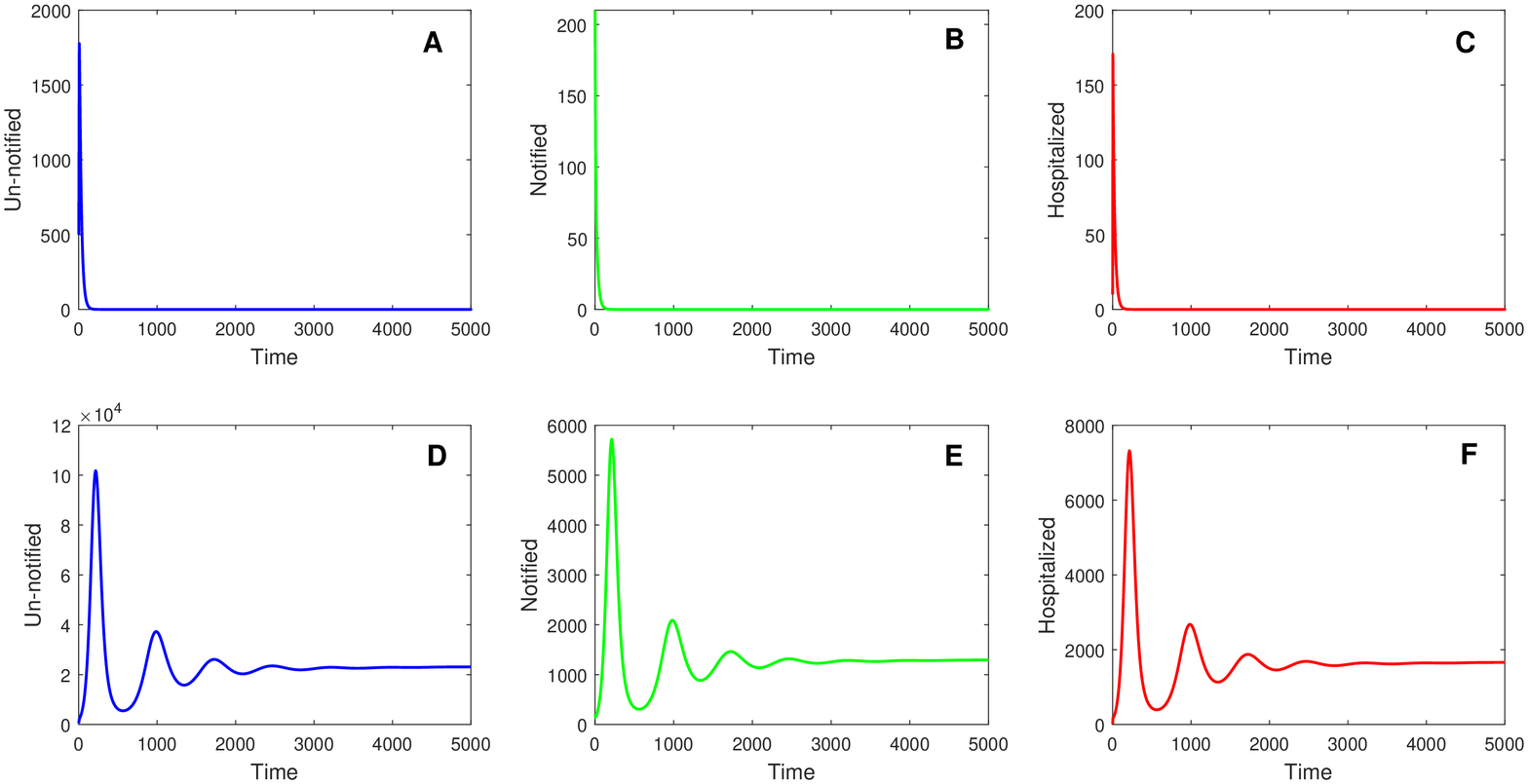}
	\end{tabular}
	\caption{Time series evaluation of infected compartments for the system (\ref{EQ:TheModel}). In Fig. (A-C), $\beta_1=0.15$ where $\mathscr{R}_0<1$, and for Fig. (D-F), $\beta_1=0.3$ where $\mathscr{R}_0>1$. The other parameter values are same as Fig. \ref{forward_bif}.}
	\label{time_series}
\end{figure}
%%%%%%%%%%%%%%%%%%%%%%%%%%%%%%%%%%%%%%
%\subsection{Global stability oef endemic equilibrium}
%The study of the endemic global stability is not only important from mathematical point of view, but also essential in predicting the evolution of the disease in the long run so that prevention and intervention strategies can be effectively designed, and public health administrative efforts can be properly scaled.

%%%%%%%%%%%%%%%%%%%%%%%%%%%%%%%%%%%%%%%%%%%%%%%%%%
\section{Model calibration}\label{Sec:calibration}
Daily new deaths between the period 1 March 2021 to 31 May 2021 for the districts Bangalore Urban and Chennai are used to calibrate the model. The rationale behind using mortality data is that this data is more reliable than the incidence data due to limited testing capacities \cite{ngonghala2020could,lau2021evaluating}. We fit the model output ($\int_{t=1}^{T} [\delta_n I_n + \delta_h I_h] dt$, where T=92) to daily new deaths and cumulative deaths due to COVID-19 for both districts. Fixed parameters of the model \eqref{EQ:TheModel} are given in Table \ref{tab:mod1}. Eight unknown model parameters are estimated, namely $\beta_1$, $\beta_2$, $\nu_u$, $\nu_h$, $\sigma_h$, $\delta_h$, $\alpha_u$ and $\alpha_n$. In addition, initial number of exposed people $E(0)$ is also estimated from the mortality data. During the specified time period, nonlinear least square solver $lsqnonlin$ (in MATLAB) is used to fit simulated daily mortality data to the reported COVID-19 deaths in Bangalore urban and Chennai districts. MATLAB codes for generating the fitting figures for Chennai has been uploaded to the GitHub repository \url{https://github.com/indrajitg-r/ODE_model_fitting}. The fitting of the daily new deaths and cumulative deaths due to COVID-19 are displayed in Fig. \ref{Fig:model_fitting_blore} and Fig. \ref{Fig:model_fitting_chennai} for Bangalore urban and Chennai, respectively. 

\begin{table}[h]
	\setlength{\tabcolsep}{12pt}
	\begin{center}		\caption{Initial conditions used to simulate the model (\ref{EQ:TheModel}) in Bangalore urban and Chennai.}
		\label{tab:ICs_blore_chennai}
		\begin{tabular}{cp{4cm}p{2cm}p{2cm}p{3cm}} \hline
			\textbf{IC's} & \textbf{Description} & \textbf{Values for Bangalore Urban} & \textbf{Values for Chennai} & \textbf{Reference} \\ \hline
			$N_{init}$ & Total population & 12765000 & 11235000 & \cite{Bangalore2021pop,Chennai2021pop}\\
			$S(0)$ & Initial number of susceptible  & $0.7 \times N_{init}$ & $0.7 \times N_{init}$ & -- \\ 
			$E(0)$ & Initial number of exposed people & (1-15000) & (1-15000) & To be estimated \\
			$I_u(0)$ & Initial number of un-notified patients & 500 & 500 & -- \\
			$I_n(0)$ & Initial number of notified patients & 210 & 189 & \cite{covid19india} \\ 
			$I_h(0)$ & Initial number of hospitalized patients & 10 & 10 & -- \\ 
			$R(0)$ & Initial number of recovered patients & 1000 & 1000 & -- \\
			$V(0)$ & Initial concentration of virus & $10^{-7}$ & $10^{-7}$ & --\\ 
			\hline
\end{tabular}\end{center}\end{table}

\begin{figure}[H]
	\centering
	\includegraphics[width=0.49\textwidth]{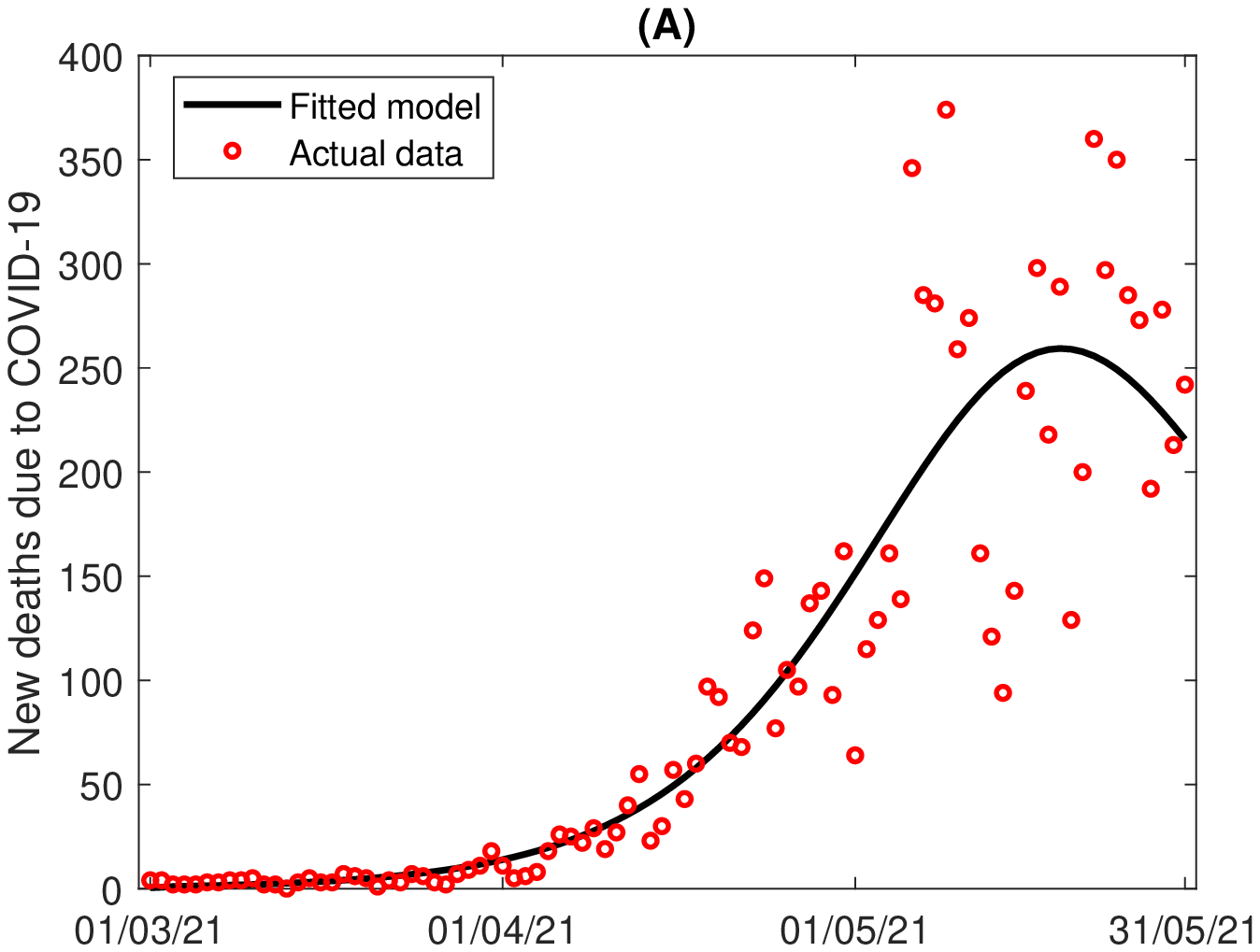}
	\includegraphics[width=0.49\textwidth]{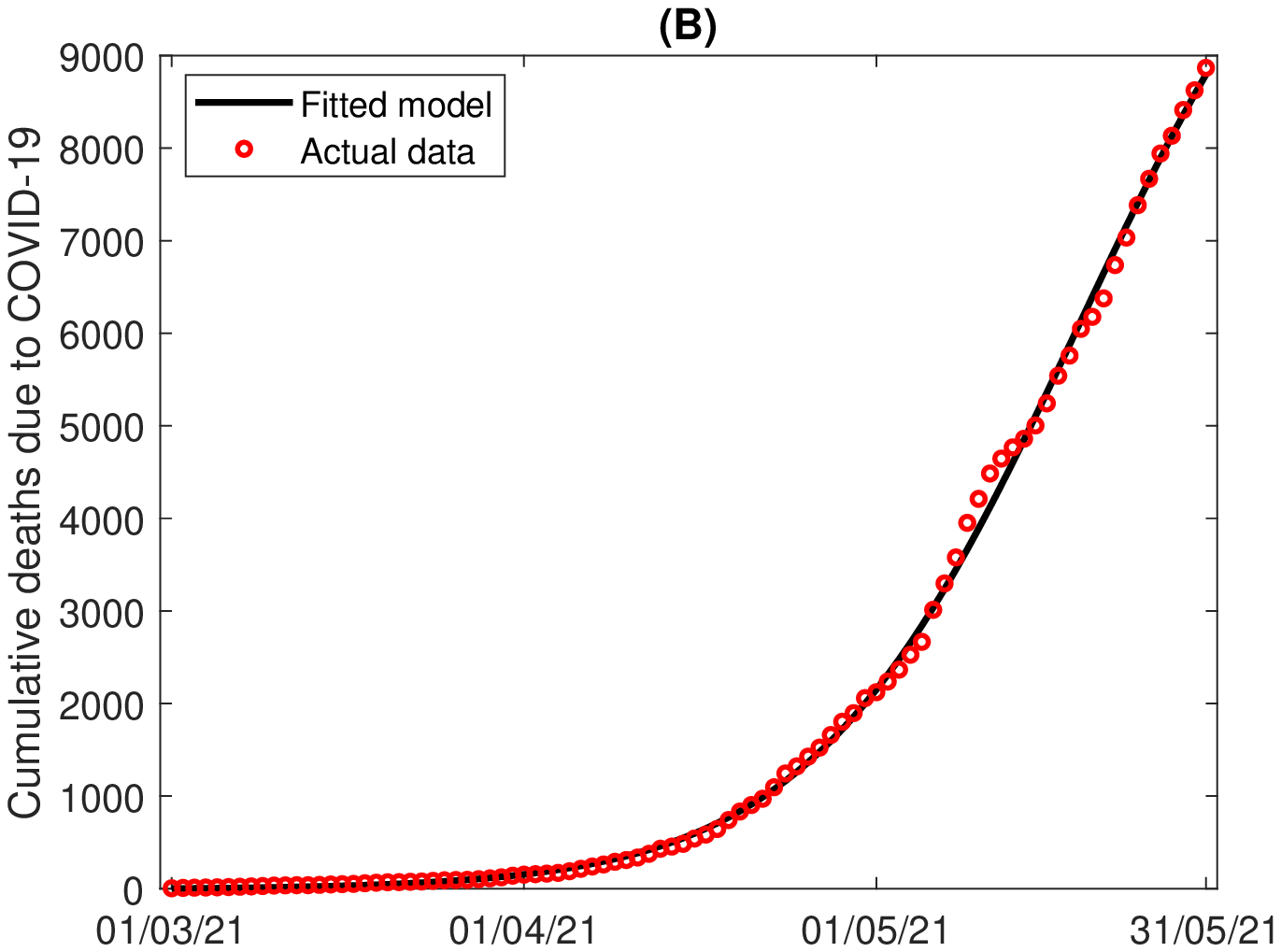}
	\caption{Fitting model solution to (A) new deaths and (B) cumulative death data due to COVID-19 in Bangalore Urban district.}
	\label{Fig:model_fitting_blore}
\end{figure}

From Fig. \ref{Fig:model_fitting_blore}, it can be argued that the model output fits the mortality data well. Estimated parameter values for Bangalore urban are as follows: $\beta_1=0.2083$, $\beta_2=6.6248 \times 10^{-6}$, $\eta_u =0.0223$, $\eta_n =0.7846$, $\sigma_h =0.9181$, $\delta_h =0.0046$, $\alpha_u =1.3232 \times 10^{-6}$, $\alpha_n = 0.0741$ and $E(0)= 4497$. The data trend is well captured by the model output as seen in Fig. \ref{Fig:model_fitting_chennai}. Estimated parameter values for Chennai district data are as follows: $\beta_1=0.2155$, $\beta_2=5.361 \times 10^{-6}$, $\eta_u =0.0012$, $\eta_n =0.5470$, $\sigma_h =0.8419$, $\delta_h =3.204 \times 10^{-4}$, $\alpha_u =6.278 \times 10^{-4}$, $\alpha_n = 0.0711$ and $E(0)= 1919$. Using these estimated parameters for both the districts, we investigate different control strategies.

%\begin{table}[h]
%	\begin{center}
%		\caption{Initial conditions used to simulate the model (\ref{EQ:TheModel}) in Chennai.}
%		\label{tab:ICs_chennai}
%		\begin{tabular}{cp{5cm}cc} \hline
%			\textbf{IC's} & \textbf{Description} & \textbf{Values} & \textbf{Reference} \\ \hline
%			$N_{init}$ & Total population of Bangalore urban & 11235000 & \citet{Chennai2021pop}\\
%			$S(0)$ & Initial number of susceptible  & $0.7 \times N_{init}$ & -- \\ 
%			$E(0)$ & Initial number of exposed people & (1-15000) & To be estimated \\
%			$I_u(0)$ & Initial number of un-notified patients & 500 & -- \\
%			$I_n(0)$ & Initial number of notified patients & 189 & \cite{covid19india} \\ 
%			$I_h(0)$ & Initial number of hospitalized patients & 10 & -- \\ 
%			$R(0)$ & Initial number of recovered patients & 1000 & -- \\
%			$V(0)$ & Initial concentration of virus & $10^{-7}$ & --\\ 
%			\hline
%\end{tabular}\end{center}\end{table}
\begin{figure}[H]
	\centering
	\includegraphics[width=0.49\textwidth]{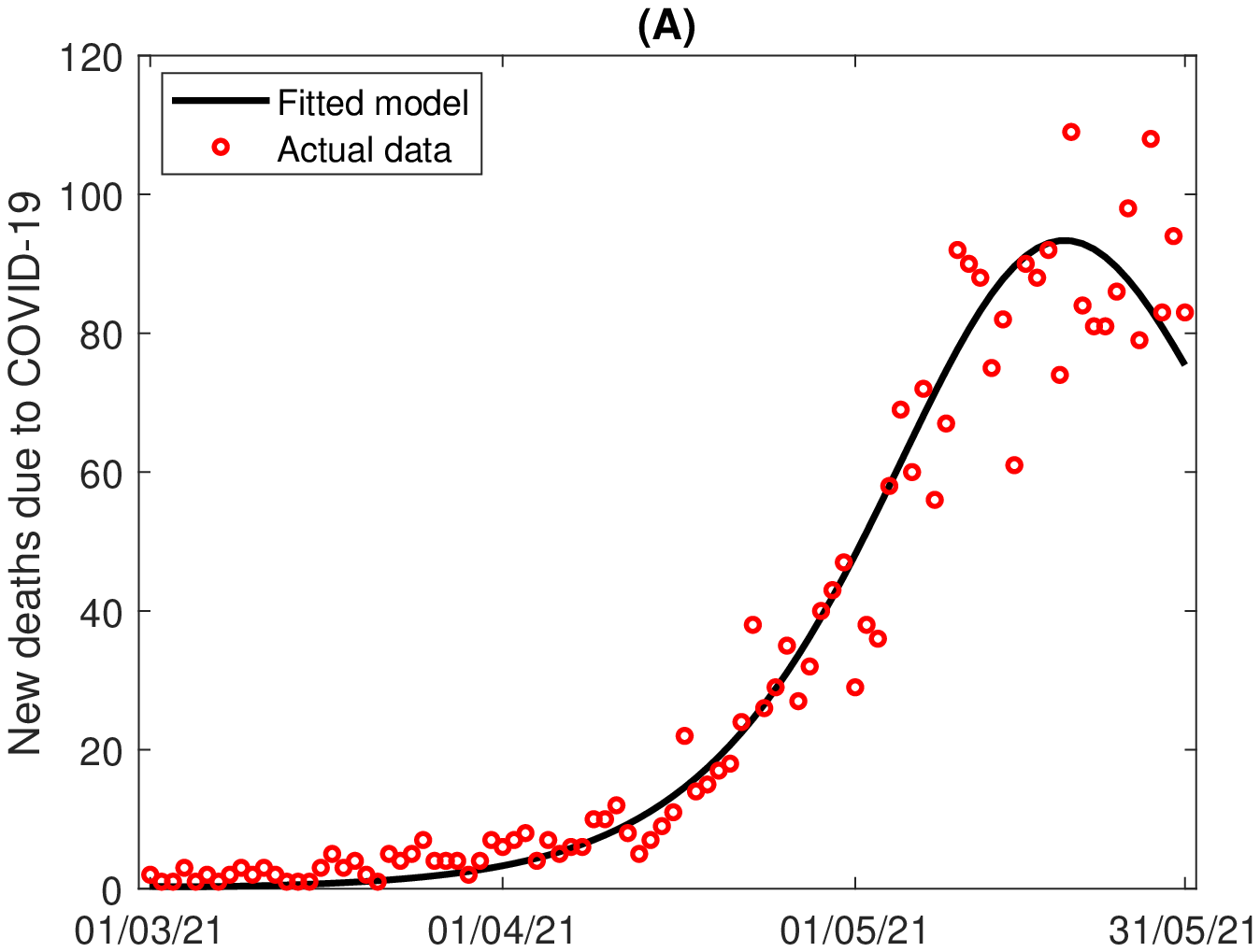}
	\includegraphics[width=0.49\textwidth]{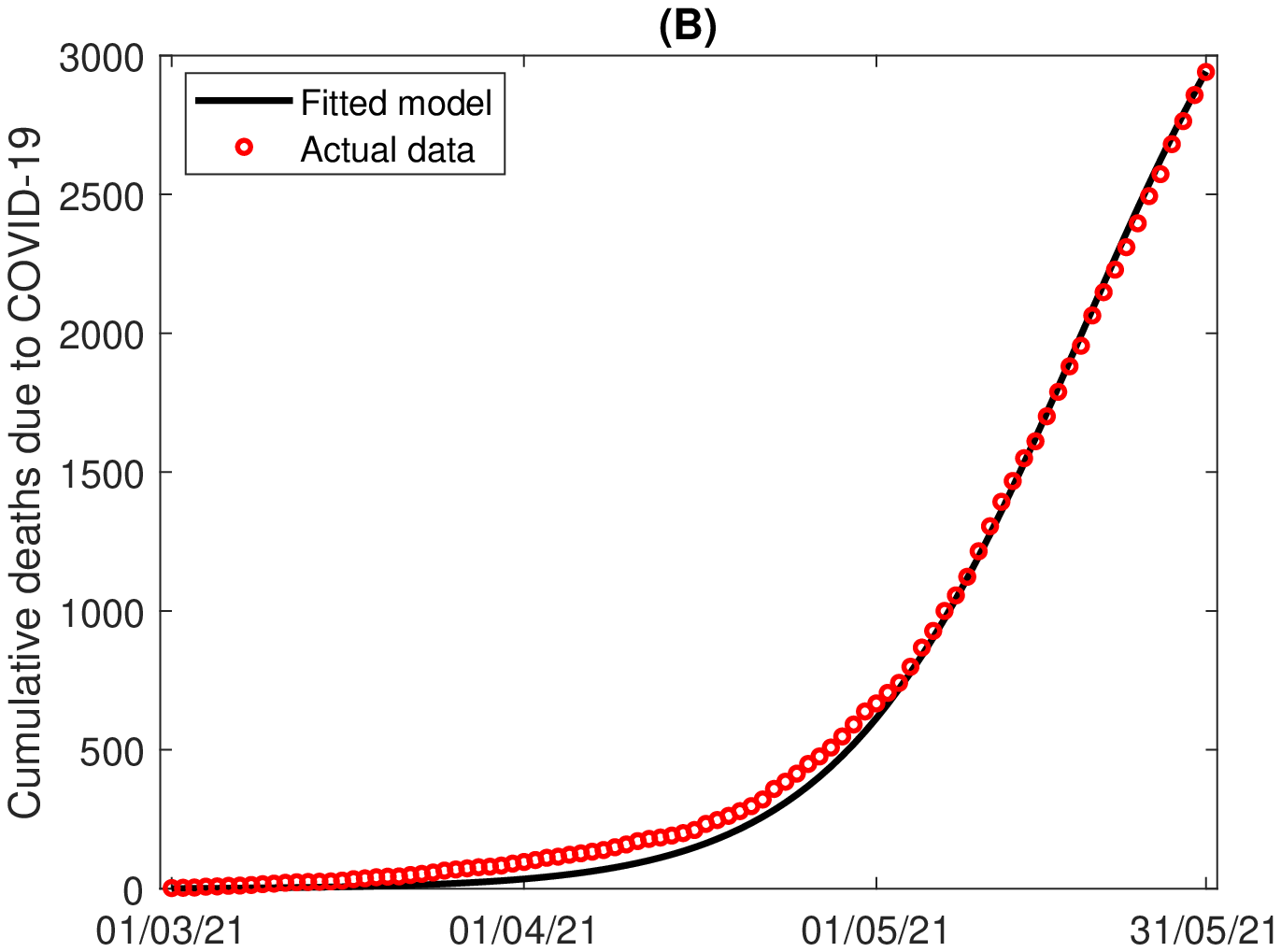}
	\caption{Fitting model solution to (A) new deaths and (B) cumulative deaths data due to COVID-19 in Chennai district.}
	\label{Fig:model_fitting_chennai}
\end{figure}

\section{Control interventions and immigration of infectives}\label{Sec:controls}
In this section, we investigate different control mechanisms and immigration of infectives through numerical simulation. We examine the effects of vaccination, treatment by drugs and use of face masks with different degrees of efficacy. 

\subsection{Use of face mask}
Although several vaccines are discovered and people are vaccinated in a rapid process, several drugs are in trial stage and few of them are implemented for hospitalized patients, nevertheless, use of face mask still could offer as a non-pharmaceutical intervention, to avoid transmission of direct contact and airborne transmission of free virus. In a densely populated country like India, it is almost impossible to determine how many people have come in contact with an infected person. Thus wearing face mask properly is an important control strategy for current epidemic outbreak. The idea  of using face mask to combat respiratory infections in community was not new \cite{van2008professional}. Face masks reduce the amount of droplet inoculum discharging from infectious individuals by capturing a proportion of droplets within the mask \cite{van2008professional, stutt2020modelling}. Face masks also reduce the amount of droplet inoculum inhaled by susceptible individuals by capturing a proportion of droplets in inhaled air and hence reducing the airborne transmission rate. Let $0\leq c_m\leq1$ is the community-wide compliance in the face mask usage and $0<\epsilon_m\leq 1$ is the face mask efficiency in preventing of the disease. Then the term $1- c_m \epsilon_m$ describe a measure of reduction in community contacts and free virus contacts due to use of face mask \cite{stutt2020modelling, ngonghala2020could}. Note that, $0\leq1- c_m \epsilon_m\leq 1$. We consider three different efficiency levels, depending on the material and layering of the face masks viz. N95 masks are 95\% effective; Surgical masks are 55\% effective and Multi-layered cloth masks are 38\% effective \cite{sharma2020efficacy}. Community-wide compliance in the face mask usage is chosen at three different levels namely, High 75\%, Medium 50\% and low 25\%. The baseline values for $c_m$ and $\epsilon_m$ are taken to be 10\% and 0.1, respectively. 

\subsection{Treatment by drugs: faster recovery}
Recently, several drugs are developed for moderate to severe COVID-19 patients and some of them are in different phase of trial period. The 3rd phase trial of Soin et al. \cite{soin2021tocilizumab} in India predict that, tocilizumab plus standard care in patients admitted to hospital with moderate to severe COVID-19 have faster recovery and reduce the burden of intensive care. Recently, Institute of Nuclear Medicine and Allied Sciences (DRDO-
INMAS) and Dr Reddy’s laboratories, Hyderabad jointly developed a drug, 2-deoxy-D-glucose (2-DG) for emergency use in symptomatic COVID-19 patients \cite{bere20212dg,samal2021anti}. On 8$^{th}$ May, 2021, the Drugs Controller General of India (DCGI) approved this drug \cite{dcgi2021pop}. As per the government release, clinical trial data shows that the drug helps in recovery timelines of hospitalized patients and reduces supplemental oxygen dependence. Consider $\epsilon_t\geq1$ be the factor describing the faster recovery of hospitalized patients. Taking the base value of the $\epsilon_t$ as 1, we vary this parameter upto 1.3 while evaluating control strategies.

\subsection{Effect of vaccination}
To control the ongoing global pandemic, vaccine is a critical tool as a pharmaceutical intervention. At present, Covaxin and Covishield vaccines are being used for the vaccination drive against COVID-19 in India  and also, several vaccine are in trial process \cite{sharun2021india}. To extend our mathematical model under the effect of vaccination, we introduced a new compartment of population, namely protected ($P(t)$), those who are vaccinated. Now, the Susceptible and recovered population both can progress to the protected population at rates $\xi_s$ and $\xi_r$, respectively, through vaccination. For simplicity, we consider that protected population completed the required doses of vaccine (and also completed 14 days after last dose) and thus does not consider any intermediate stage (i.e., the time gap between doses or before 14 days of last dose). The base values for vaccination rates are chosen to be 0.00085 for both susceptibles and recovered people. This value is calculated from the vaccination data of India \cite{covid19india}. Simulations are performed by increasing the vaccination rates to 0.05 for both the sub-populations.

\subsection{Immigration of infectives}
Communicable disease like the current COVID-19 epidemic may be introduced into a community by the arrival of infectives outside the community. Indeed, in India, the first COVID-19 case was reported in Trissur, Kerala, on January 30, 2020, who was returned from Wuhan, China \cite{andrews2020first}. Thus, immigration of infectives in a disease-free community played a crucial role for the spread of the disease. We extend the mathematical model (\ref{EQ:TheModel}) by assuming a constant flow of new members of un-notified person. Recruitment of notified person was not considered as they are restricted for travelling. On the other hand, immigration is necessary for many people to maintain a livelihood. Thus, individuals from neighboring locations must be allowed to immigrate for the well being of the overall community. Let $\Lambda$ be the recruitment rate of un-notified infected individuals through immigration. We consider three levels of immigration as low, medium and high with $\Lambda$ = 100, 1000 and 5000 respectively. 

By considering above control interventions and immigration of infectives, the model (\ref{EQ:TheModel}) can be extended to the following model:
\begin{eqnarray}
	\begin{cases}
		\begin{array}{lll}
			\displaystyle{\od{S}{t}} &=& \Pi-\frac{\beta_1 (1- c_m \epsilon_m) S( I_u+ \nu I_n)}{N} - \beta_2 (1- c_m \epsilon_m) SV -(\mu + \xi_s) S + \theta R, \\
			\displaystyle{\od{E}{t}} &=& \frac{\beta_1 (1- c_m \epsilon_m) S( I_u+ \nu I_n)}{N} + \beta_2 (1- c_m \epsilon_m) SV - (\gamma + \mu) E, \\
			\displaystyle{\od{I_u}{t}} &=& \Lambda + (1-p) \gamma E - (\eta_u + \sigma_u +\mu )I_u, \\
			\displaystyle{\od{I_n}{t}} &=&  p \gamma E-(\eta_n + \sigma_n + \mu + \delta_n) I_n, \\
			\displaystyle{\od{I_h}{t}} &=& \eta_u I_u+ \eta_n I_n - ( \epsilon_t \sigma_h + \mu + \delta_h)I_h,  \\
			\displaystyle{\od{R}{t}} &=& \sigma_u I_u+\sigma_n I_n + \epsilon_t \sigma_h I_h - (\theta + \mu +\xi_r)R, \\
			\displaystyle{\od{P}{t}} &=& \xi_s S + \xi_r R - \mu P, \\
			\displaystyle{\od{V}{t}} &=& \alpha_u I_u+ \alpha_n I_n - \mu_c V.
		\end{array}
		\label{EQ:TheModel_with_control}
	\end{cases}
\end{eqnarray}

Intensive numerical simulations are performed to quantify the effects of various strategies in Bangalore urban and Chennai district. The fixed parameters and initial conditions are taken from Table \ref{tab:mod1} and \ref{tab:ICs_blore_chennai}. The initial number of protected/vaccinated people is estimated using vaccination coverage data from both the districts \cite{covid19india}. Estimated parameters for both districts are mentioned in the model calibration section. Using these parameters and control parameters at various levels we simulate the model (\ref{EQ:TheModel_with_control}) with three months ahead projections. 

For both the districts, the projection period is 1 June 2021 to 31 August 2021. Initially, we simulate the effects of single control parameters while other parameters are fixed at the base level. Change in the hospitalized populations with different control levels are depicted for Bangalore urban and Chennai in Fig. \ref{Fig:single_control_Bangalore} and Fig. \ref{Fig:single_control_chennai}, respectively. For both the districts it can be seen that $\epsilon_t$ and $\xi_s$ have significant impact on reduction of hospitalized persons. Face mask related parameters ($\epsilon_m$ and $c_m$) and vaccination rate of recovered persons ($\xi_r$) show minimal effect on reduction of hospitalized populations. On the other hand, immigration rate of un-notified infected people has significant on the hospitalized populations. In particular, high immigration rates may drive the hospitalized cases to high values.

Further, to quantify the effects of single control strategies more precisely, we calculate the percentage reduction of notified and hospitalized persons in the three month projection period. We use the following basic formula
\begin{eqnarray}\label{percent_reduction}
	\textmd{Percentage reduction}=\frac{\textmd{Baseline cases} -
		\textmd{Cases with control}}{\textmd{Baseline cases}}\times 100.
\end{eqnarray}

The percentage reduction in notified and hospitalized cases are reported in Table \ref{percent_reduction_single_control} for Bangalore urban and Chennai districts. The fixed parameters and initial conditions are taken from Tables \ref{tab:mod1} and \ref{tab:ICs_blore_chennai}.

\begin{figure}[H]
	\centering
	\includegraphics[width=0.49\textwidth]{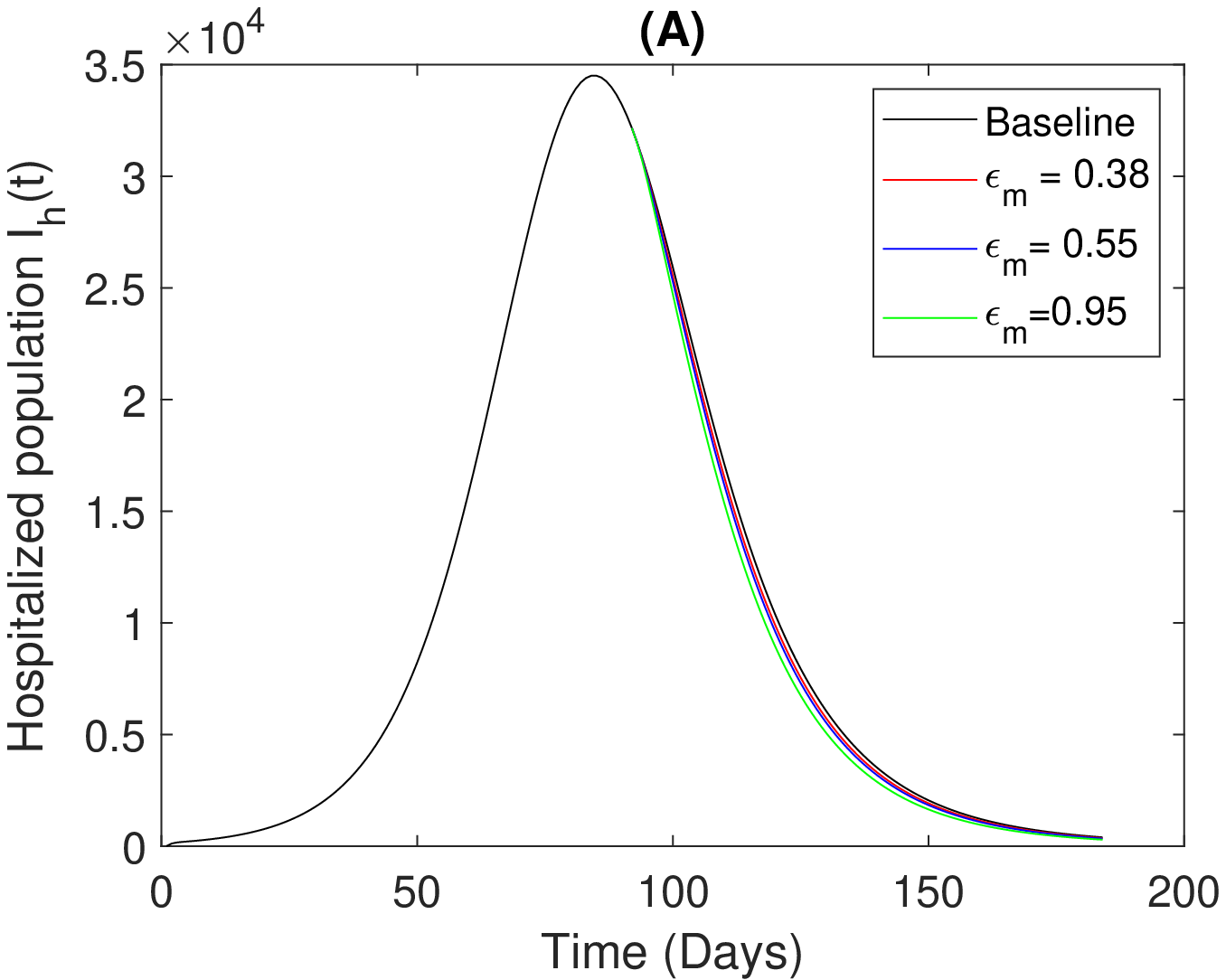}
	\includegraphics[width=0.49\textwidth]{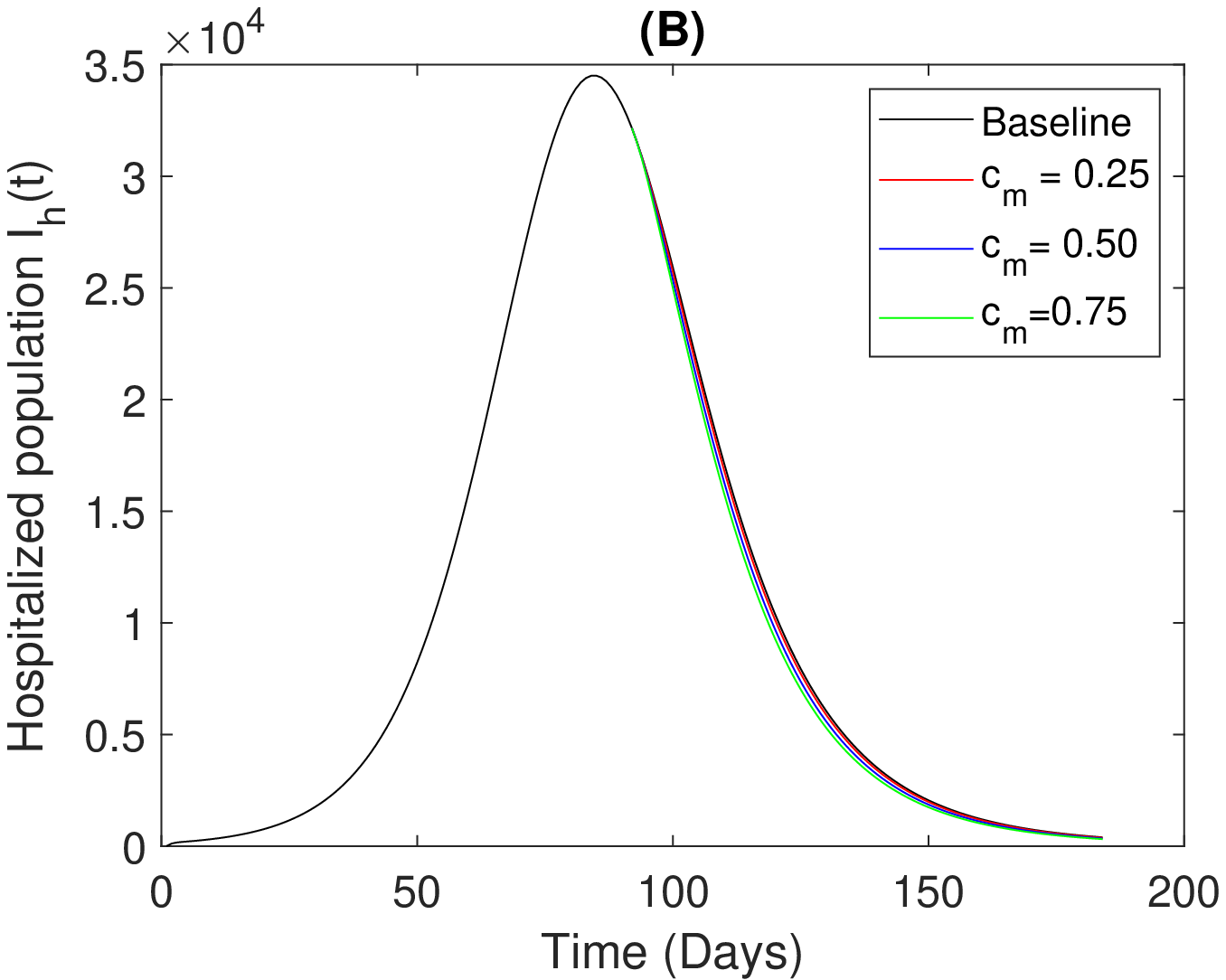}\\
	\includegraphics[width=0.49\textwidth]{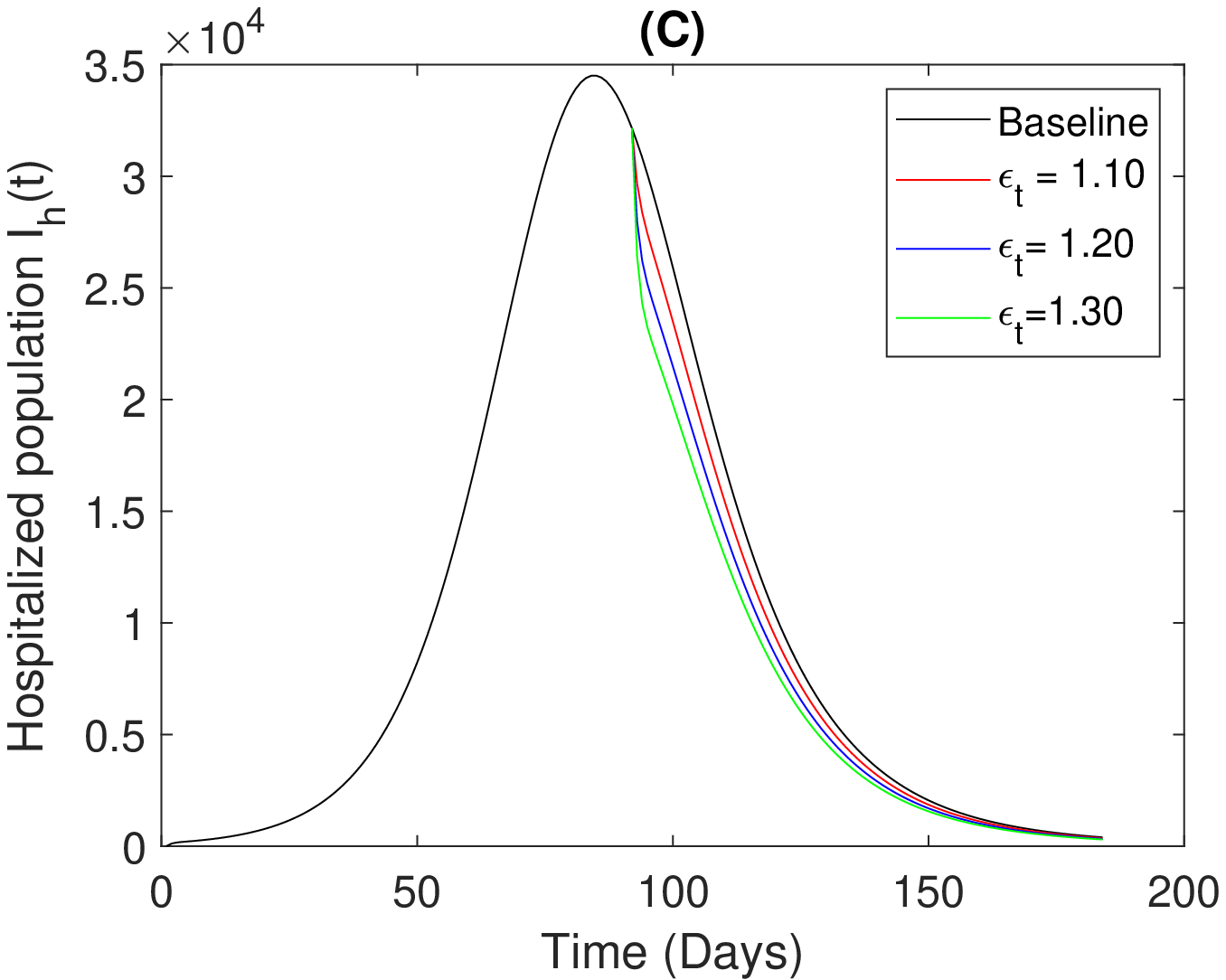}
	\includegraphics[width=0.49\textwidth]{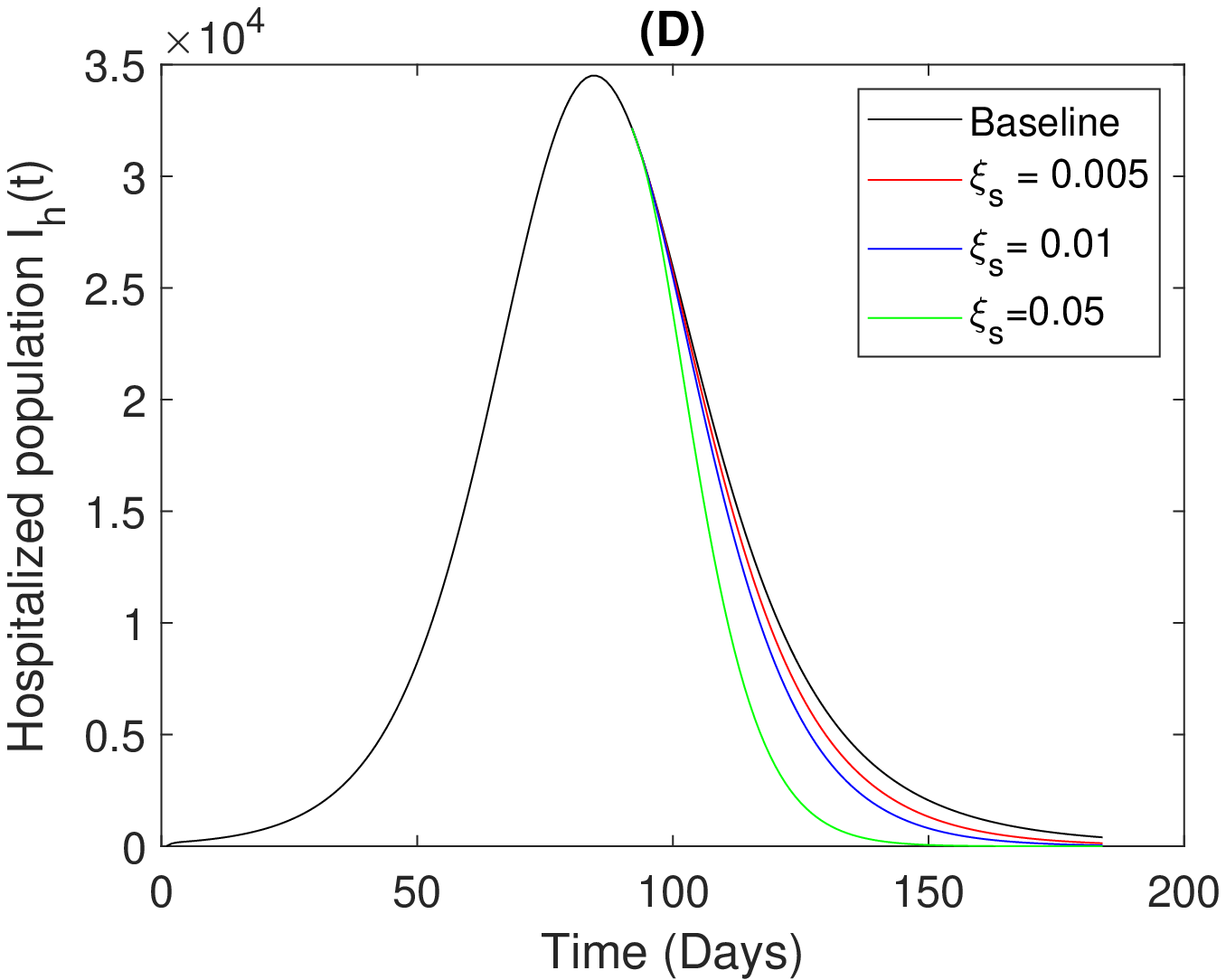}\\
	\includegraphics[width=0.49\textwidth]{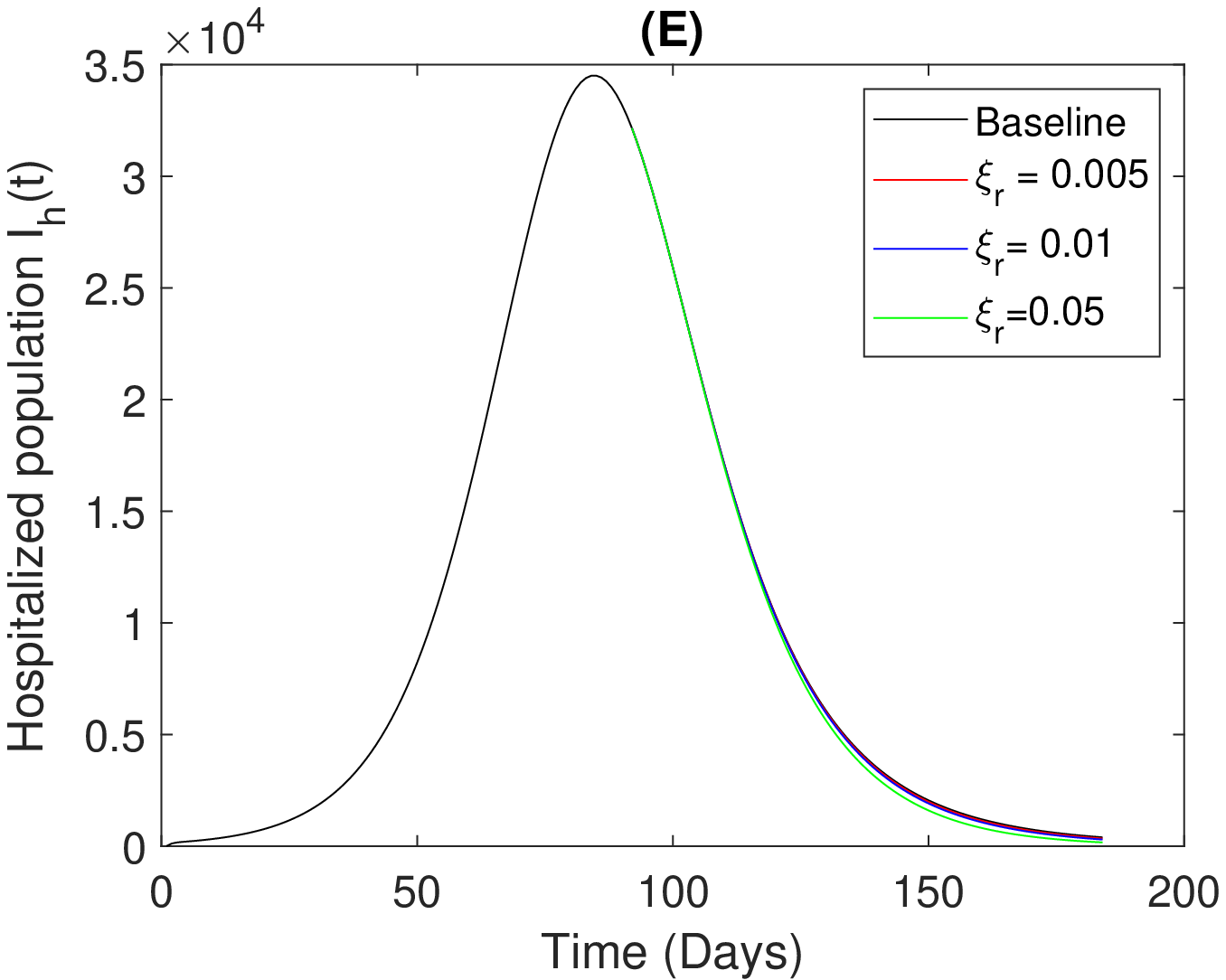}
	\includegraphics[width=0.49\textwidth]{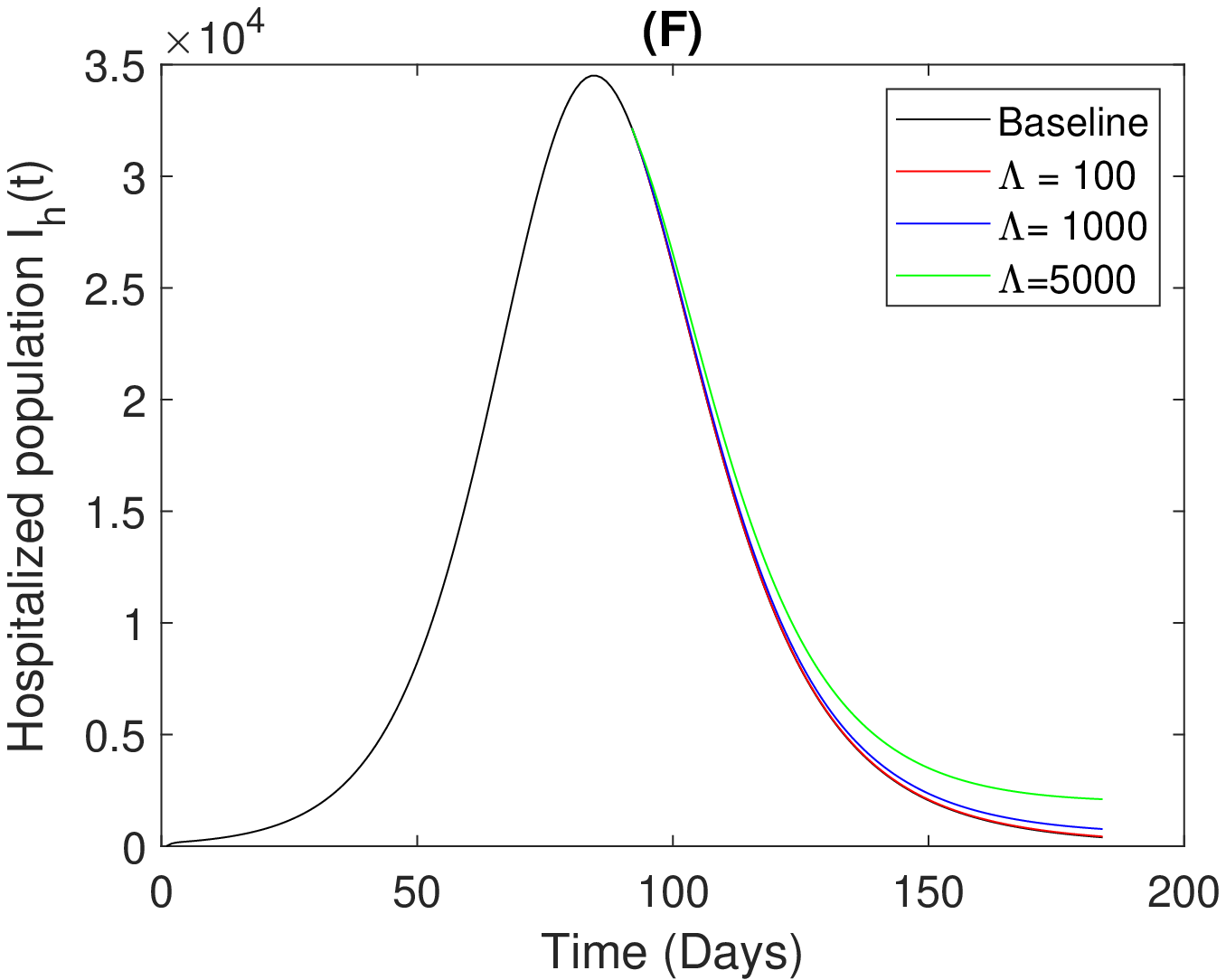}
	\caption{Effects of single control strategies in Bangalore urban namely (A) efficacy of face mask usage, (B) community-wide compliance in face mask usage, (C) increase in recovery rate of hospitalized patients, (D) vaccination rate of susceptible individuals, (E) vaccination rate recovered people and (F) immigration of infectives.}
	\label{Fig:single_control_Bangalore}
\end{figure}

\begin{figure}[H]
	\centering
	\includegraphics[width=0.49\textwidth]{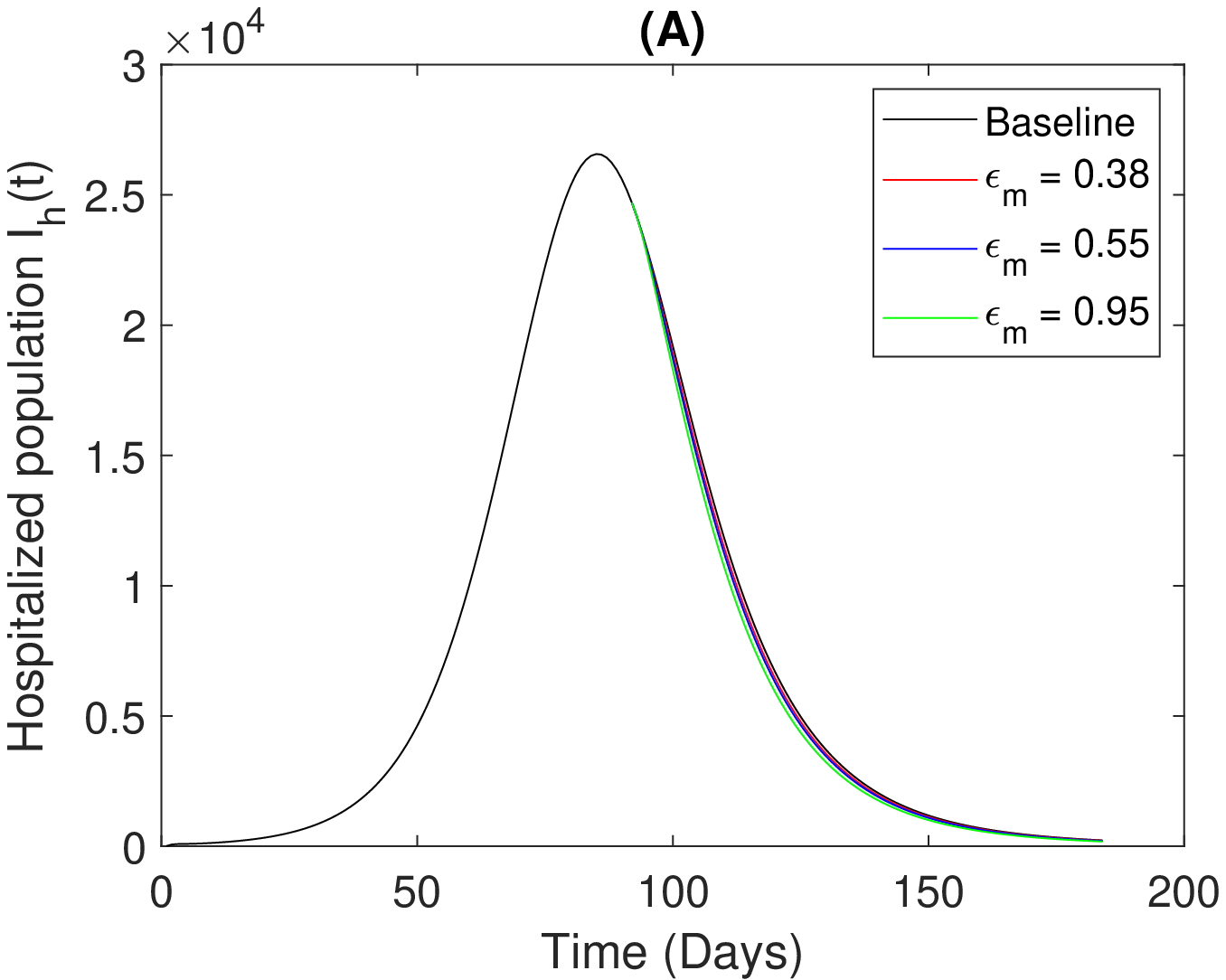}
	\includegraphics[width=0.49\textwidth]{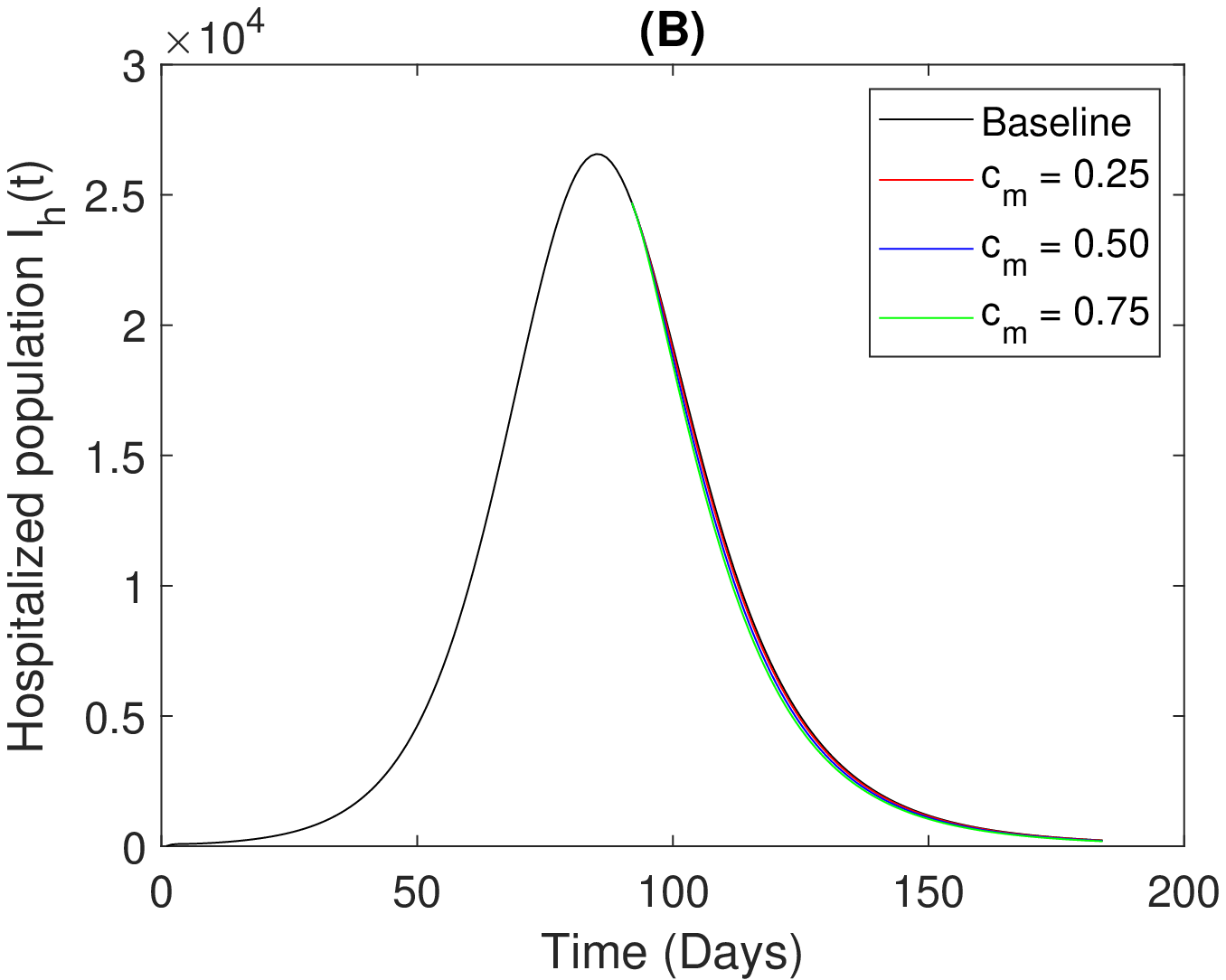}\\
	\includegraphics[width=0.49\textwidth]{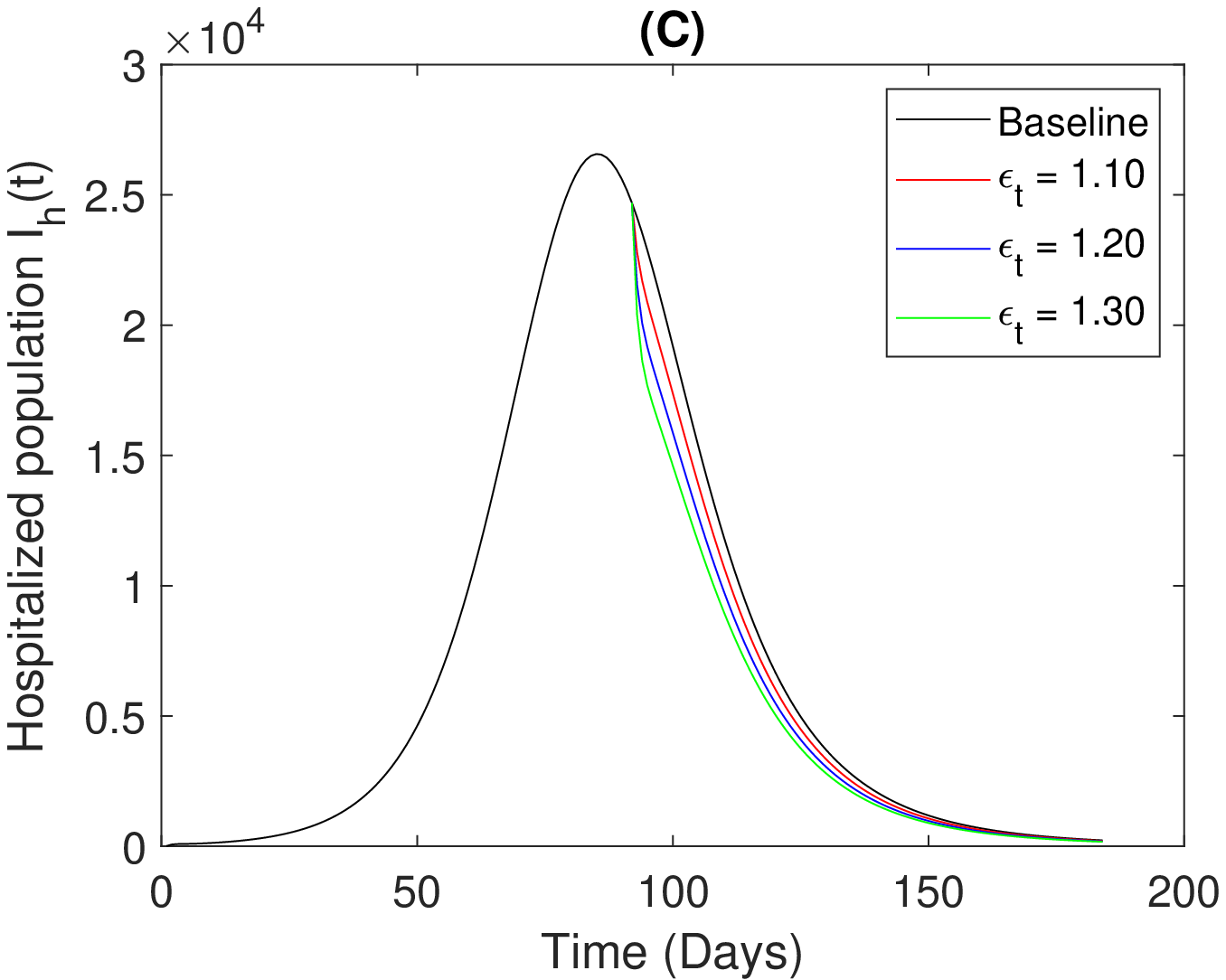}
	\includegraphics[width=0.49\textwidth]{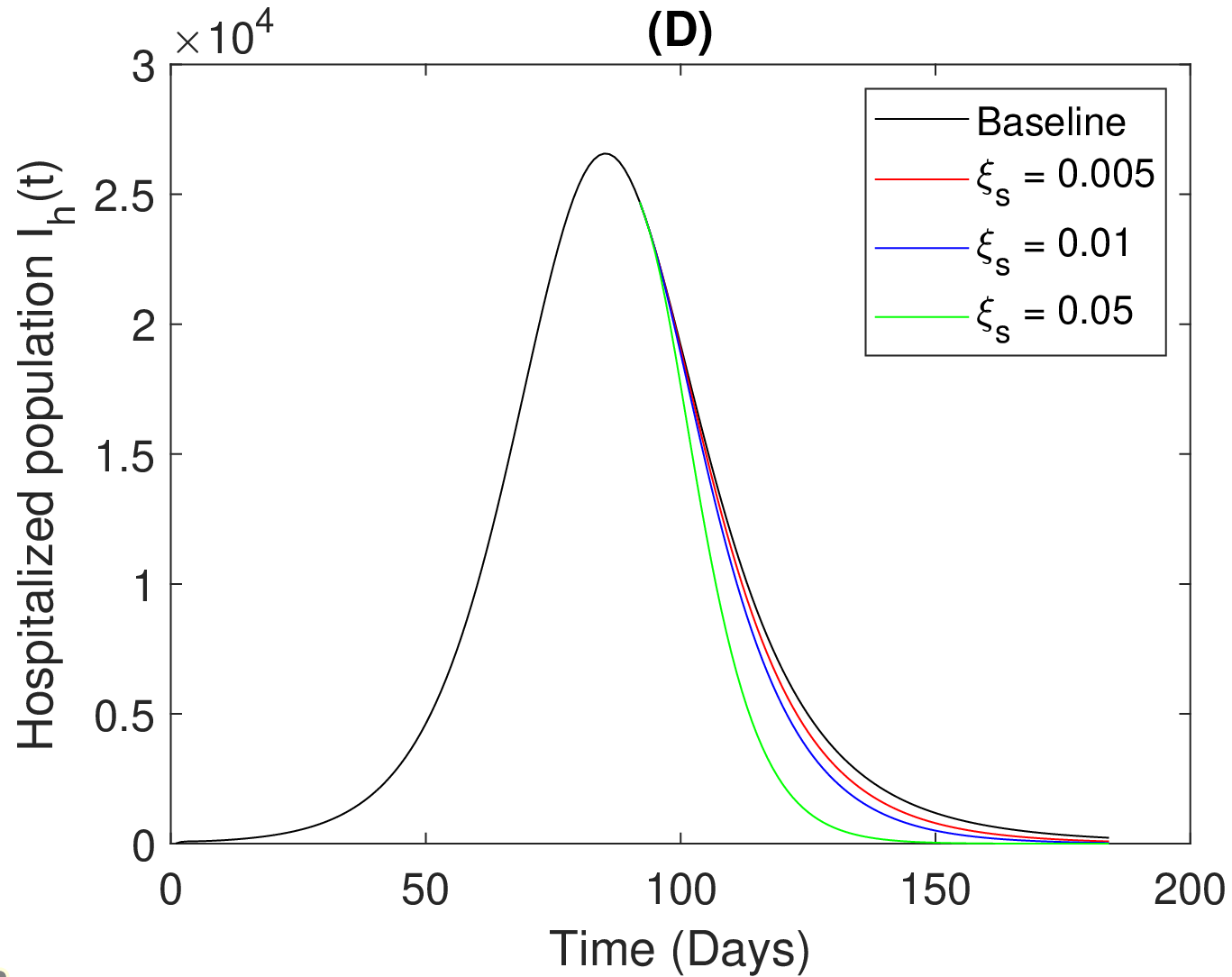}\\
	\includegraphics[width=0.49\textwidth]{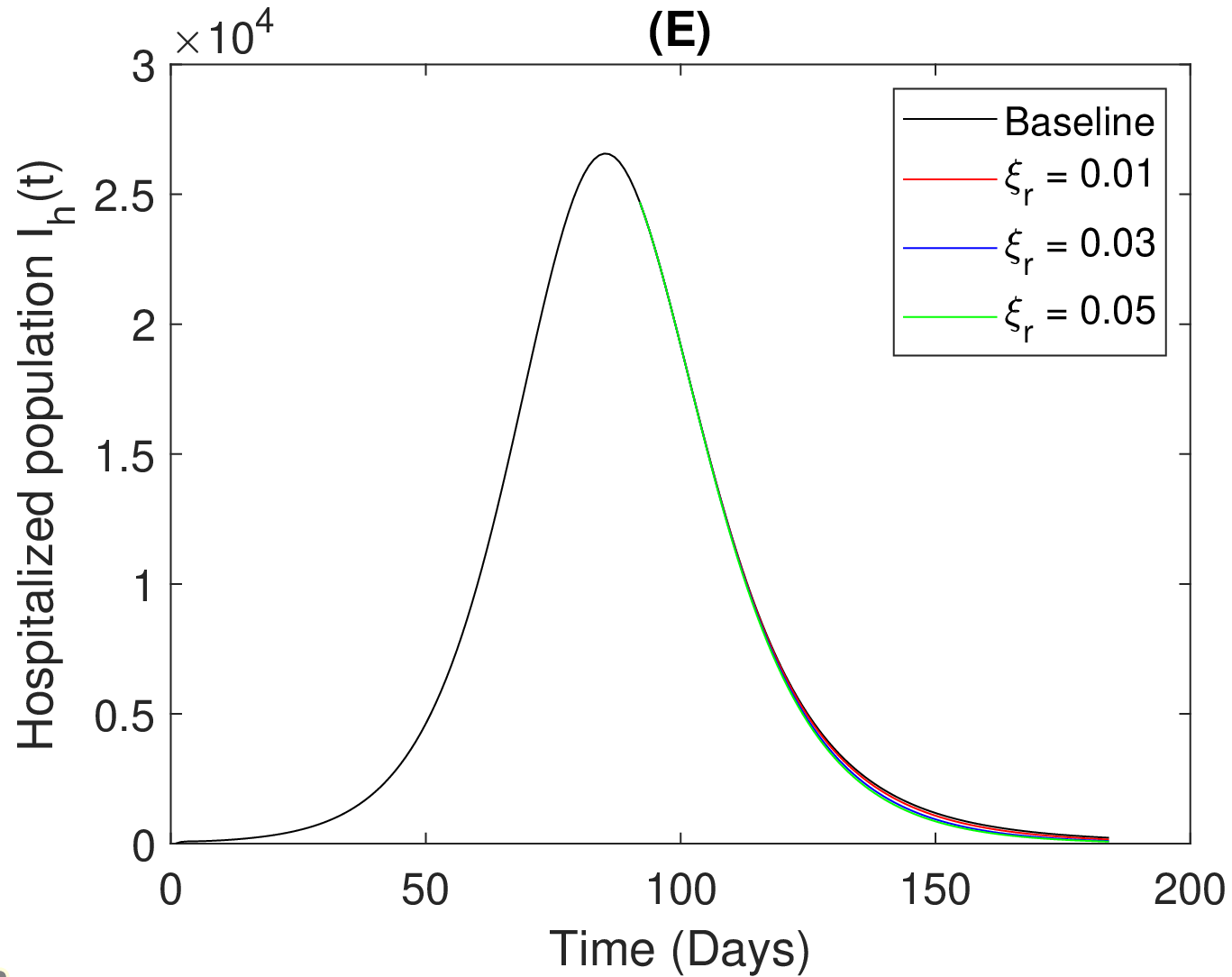}
	\includegraphics[width=0.49\textwidth]{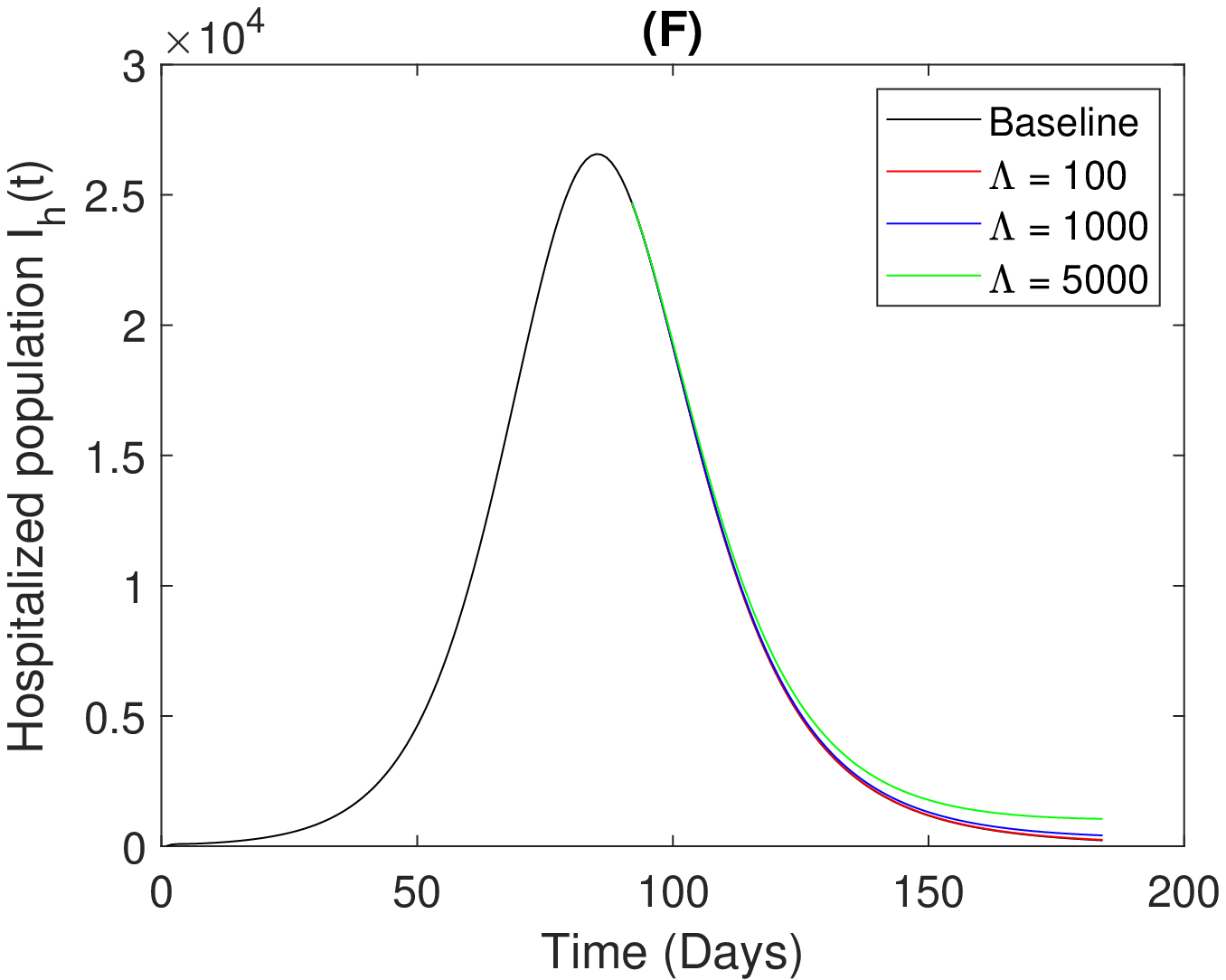}
	\caption{Effects of single control strategies in Chennai namely (A) efficacy of face mask usage, (B) community-wide compliance in face mask usage, (C) increase in recovery rate of hospitalized patients, (D) vaccination rate of susceptible individuals, (E) vaccination rate recovered people and (F) immigration of infectives.}
	\label{Fig:single_control_chennai}
\end{figure}

We can see that the face mask related controls when applied individually have moderate impact on the notified and hospitalized populations of both locations. On the other hand, treatment of hospitalized population has large positive impact on the percentage reduction of hospitalized patients. However, treatment do not show any impact on the notified patients for both the locations. This might be due to the fact that the hospitalized patients are not responsible for new infections directly. Comparing the vaccination rates of susceptible and recovered people, it can be inferred from the table that vaccinating susceptibles has better impact than that of recovered people. At least in the short-term (three months) scenario, it is beneficial to vaccinate susceptible people than to vaccinate recovered people. Furthermore, while evaluating the effects of immigration, we noticed that high level of immigration has significant negative impact on the percentage reduction. Thus, the immigration of un-notified COVID-19 patients has to be as low as possible. This can be done by allowing people from low incidence neighbor districts and restricting the immigration from highly affected areas. From the single control results it is evident that pharmaceutical mitigation strategies like treatment of hospitalized patients and vaccination of susceptible individuals are most effective. Additionally, non-pharmaceutical measures such as face mask use has minimal effects on the reduction of notified and hospitalized COVID-19 cases. 

\begin{table}[H]
	\centering \caption{Percentage reduction in total number of notified and hospitalized COVID-19 patients for different levels of interventions.}\label{percent_reduction_single_control}
	\begin{tabular}{|c|c|c|c|c|}
		\hline
		\multirow{2}{*}{Parameter values} & \multicolumn{2}{c|}{Bangalore urban} & \multicolumn{2}{c|}{Chennai}
		\\ \cline{2-3} \cline{4-5} & Reduction in $I_n$ & Reduction in $I_h$  & Reduction in $I_n$ & Reduction in $I_h$ \\ \hline
		$\epsilon_m$ = 0.38  & 3.25 & 2.91 & 2.50 & 2.35 \\ 
		~~~~~= 0.55	 & 5.21 & 4.67 & 4.03 & 3.78 \\ 
		~~~~~= 0.95  & 9.79 & 8.77 & 7.65 & 7.18 \\ \hline
		$c_m$ = 0.25  & 1.74 & 1.56 & 1.33 & 1.25 \\ 
		~~~~~= 0.50  & 4.64 & 4.15 & 3.58 & 3.36 \\ 
		~~~~~= 0.75  & 7.51 & 6.73 & 5.83 & 5.48 \\ \hline
		$\epsilon_t$ = 1.10  & -0.00 & 8.83 & -0.00 & 8.84 \\ 
		~~~~~= 1.20  & -0.00 & 16.18 & -0.00 & 16.21 \\ 
		~~~~~= 1.30  & -0.00 & 22.41 & -0.00 & 22.44 \\ \hline
		$\xi_s$ = 0.005 & 8.67 & 7.71 & 6.93 & 6.48 \\ 
		~~~~~= 0.01  & 15.91 & 14.20 & 13.06 & 12.24 \\ 
		~~~~~= 0.05  & 39.42 & 35.33 & 34.85 & 32.72 \\ \hline
		$\xi_r$ = 0.01  & 0.61 & 0.53 & 1.24 & 1.15 \\ 
		~~~~~= 0.03  & 1.23 & 1.08 & 2.97 & 2.76 \\ 
		~~~~~= 0.05  & 3.98 & 3.51 & 4.02 & 3.74 \\ \hline
		$\Lambda$ = 100 & -0.21 & -0.32 & -0.19 & -0.19 \\ 
		~~~~~= 1000  & -2.03 & -3.12 & -1.90 & -1.87 \\ 
		~~~~~= 5000  & -9.41 & -15.01 & -8.76 & -8.70 \\ \hline
	\end{tabular}
\end{table}

To investigate the combined effects of control interventions, we examine four different combinations on the total number of hospitalized patients ($I_h^{tot}$) in the three month projection period. The parameter values are taken as mentioned earlier in this section. We draw contour plots with respect to the control parameters with response variable $I_h^{tot}=\int_{t=92}^{184} I_h(t) dt$. Contour plots depicting combination of interventions are displayed in Fig. \ref{Fig:two_controls_Bangalore} and Fig. \ref{Fig:two_controls_Chennai} for Bangalore urban and Chennai districts, respectively. 

\begin{figure}[h]
	\centering
	\includegraphics[width=0.45\textwidth]{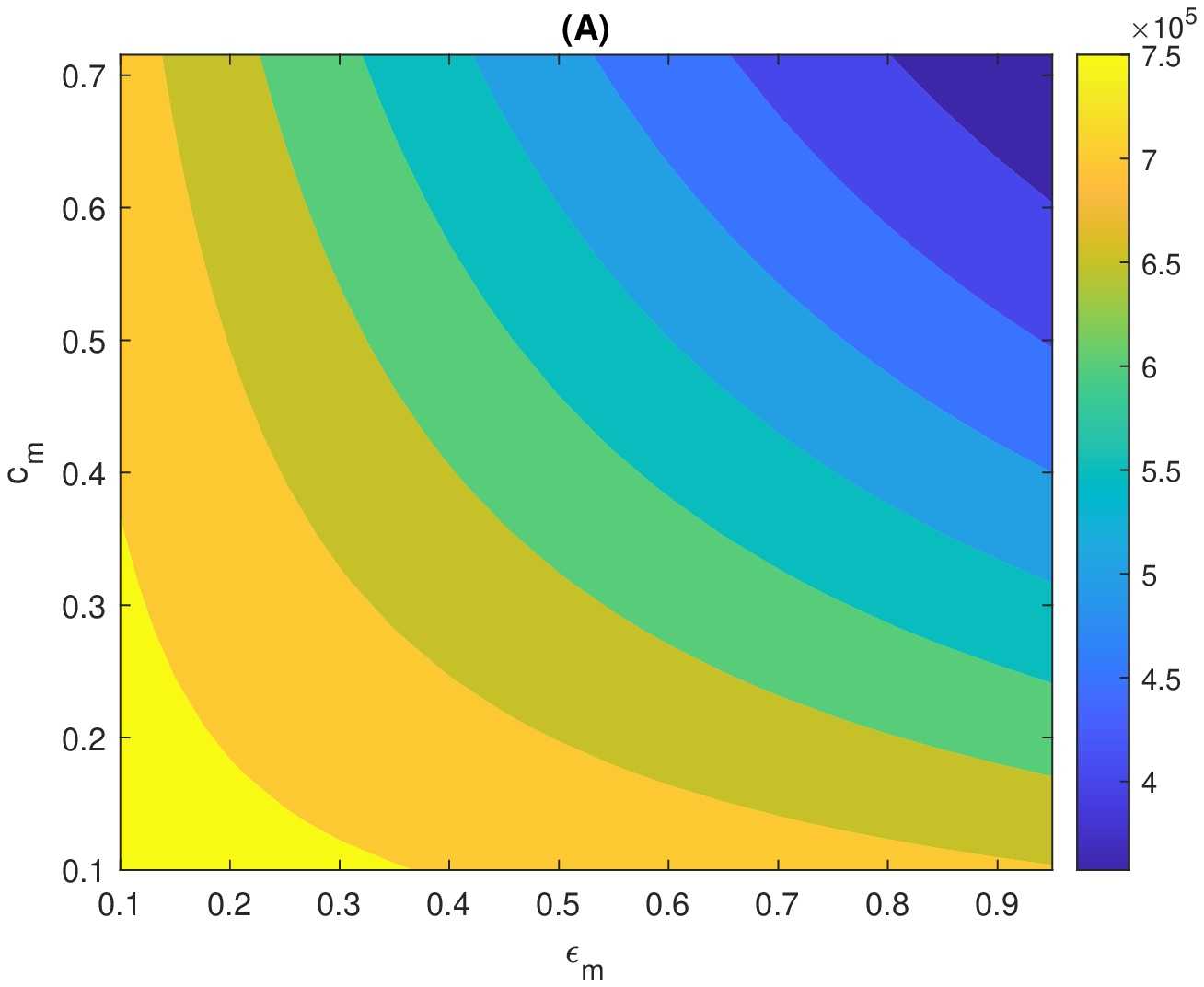}
	\includegraphics[width=0.45\textwidth]{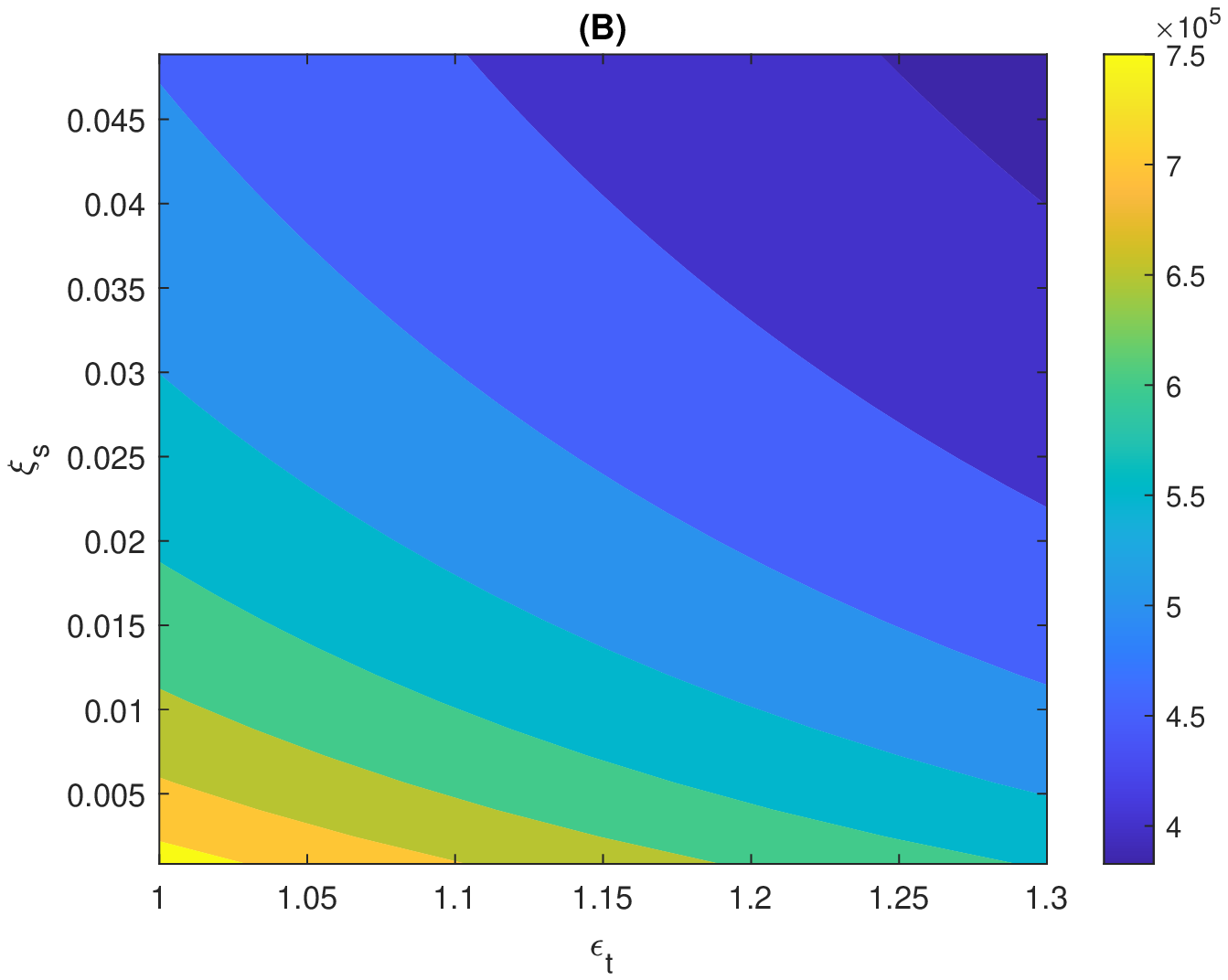}
	\includegraphics[width=0.45\textwidth]{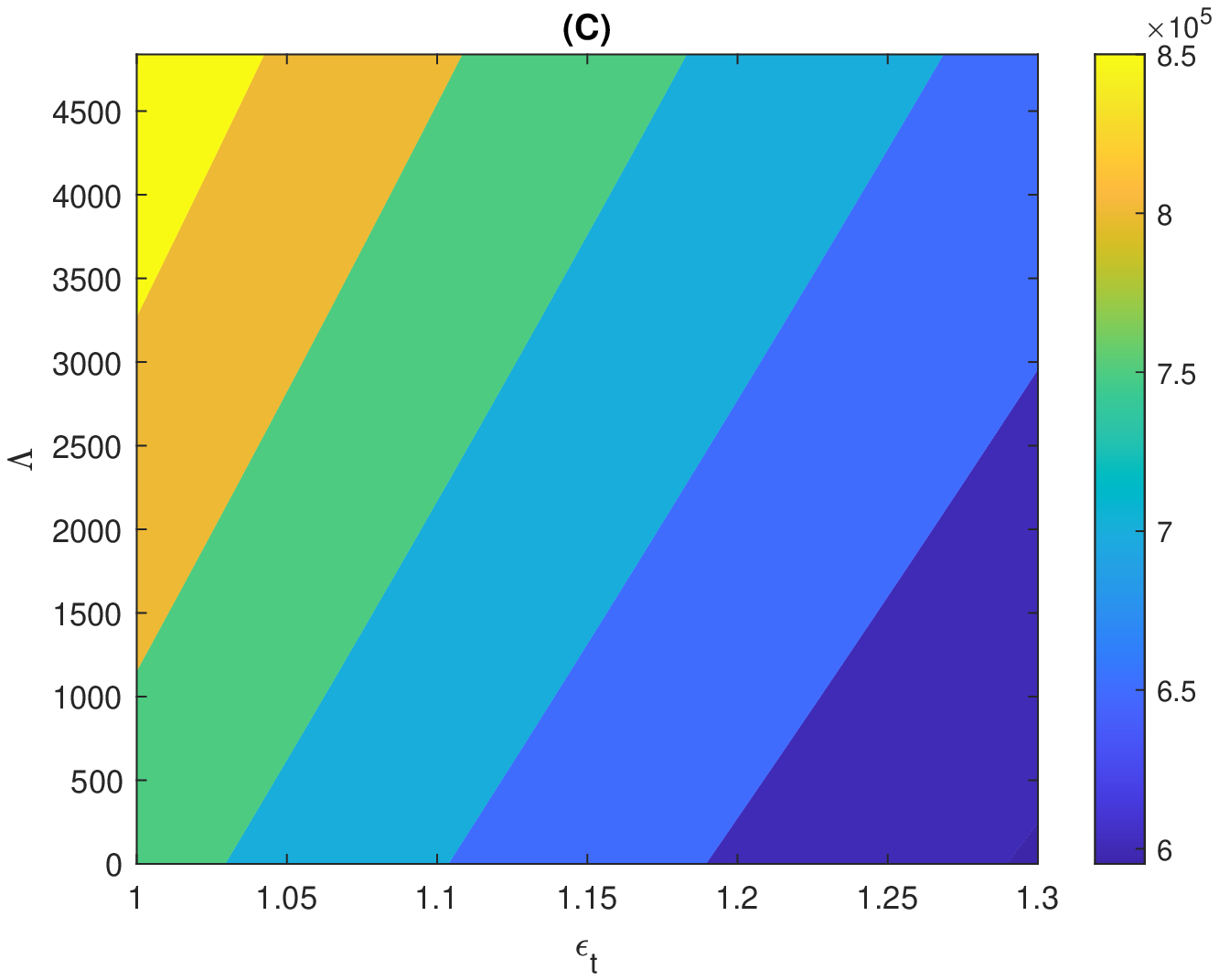}
	\includegraphics[width=0.45\textwidth]{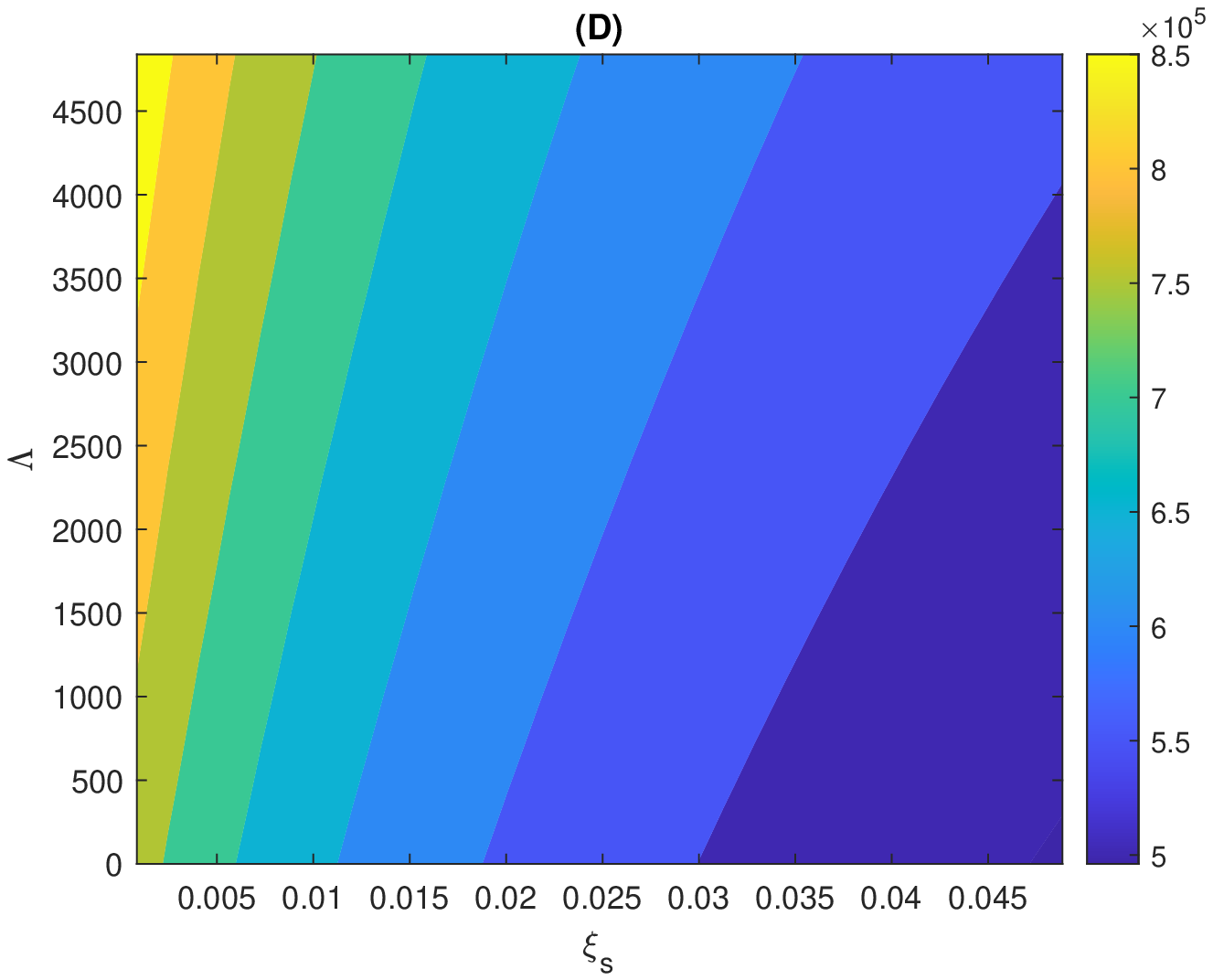}
	\caption{Combination of control strategies and immigration of infectives in Bangalore urban. (A)  efficacy of face mask usage  $ \textendash$ community-wide compliance in face mask usage ($\epsilon_m-c_m$), (B) increase in recovery rate of hospitalized patients $ \textendash$ vaccination rate of susceptible individuals ($\epsilon_t-\xi_s$), (C) increase in recovery rate of hospitalized patients $ \textendash$ immigration of infectives ($\epsilon_t-\Lambda$), (D) vaccination rate of susceptible individuals $ \textendash$ immigration of infectives ($\epsilon_s-\Lambda$).}
	\label{Fig:two_controls_Bangalore}
\end{figure}

For both of the locations, similar scenarios are seen with varying control interventions and immigration of un-notified infectious patients. Simultaneously increasing face mask efficacy and population-wide compliance level will decrease $I_h^{tot}$ in the three month projection period ( Fig. \ref{Fig:two_controls_Bangalore}(A) and \ref{Fig:two_controls_Chennai}(A)). Thus both the face mask related parameters will have stronger impact on the reduction of $I_h^{tot}$. On the other hand, vaccination rate of susceptibles and increase in the treatment rate of hospitalized patients will have similar effects for both the locations (as seen in Fig. \ref{Fig:two_controls_Bangalore}(B) and \ref{Fig:two_controls_Chennai}(B)). Further, the immigration of un-notified COVID-19 patients will obviously show increase in $I_h^{tot}$.
But keeping $\Lambda$ low and increasing treatment rate simultaneously will keep $I_h^{tot}$ under control (see Fig. \ref{Fig:two_controls_Bangalore}(C) and \ref{Fig:two_controls_Chennai}(C)). From the number of $I_h^{tot}$, it should be noted that this combination is risky if there is no other intervention in action. However, if the vaccination of susceptibles are performed with a increased rate, then the immigration may not have very adverse effects (see Fig. \ref{Fig:two_controls_Bangalore}(D) and \ref{Fig:two_controls_Chennai}(D)). It can be noted that high level of face mask use with maximum efficacy is more effective than the other combined strategies in terms of reduction in $I_h^{tot}$. But the feasibility of face mask usage in such intensity may be critical. From the contour plot analysis, it can be argued that the combination of $\xi_s$ and $\epsilon_t$ is more feasible and has competitive effects on $I_h^{tot}$. Further numerical simulations are necessary to study combination of more than two control interventions and immigration of un-notified infectives.

\begin{figure}[h]
	\centering
	\includegraphics[width=0.45\textwidth]{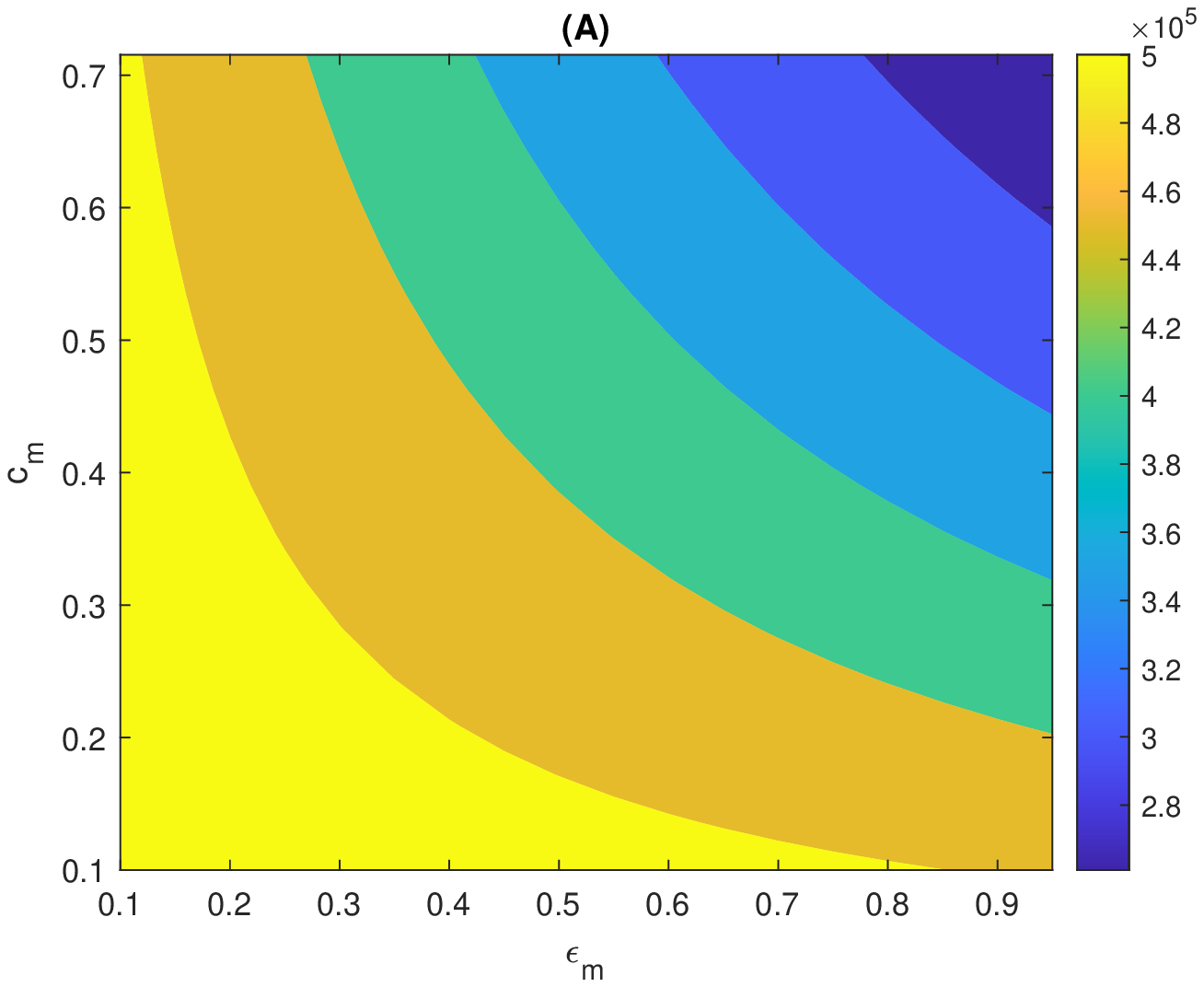}
	\includegraphics[width=0.45\textwidth]{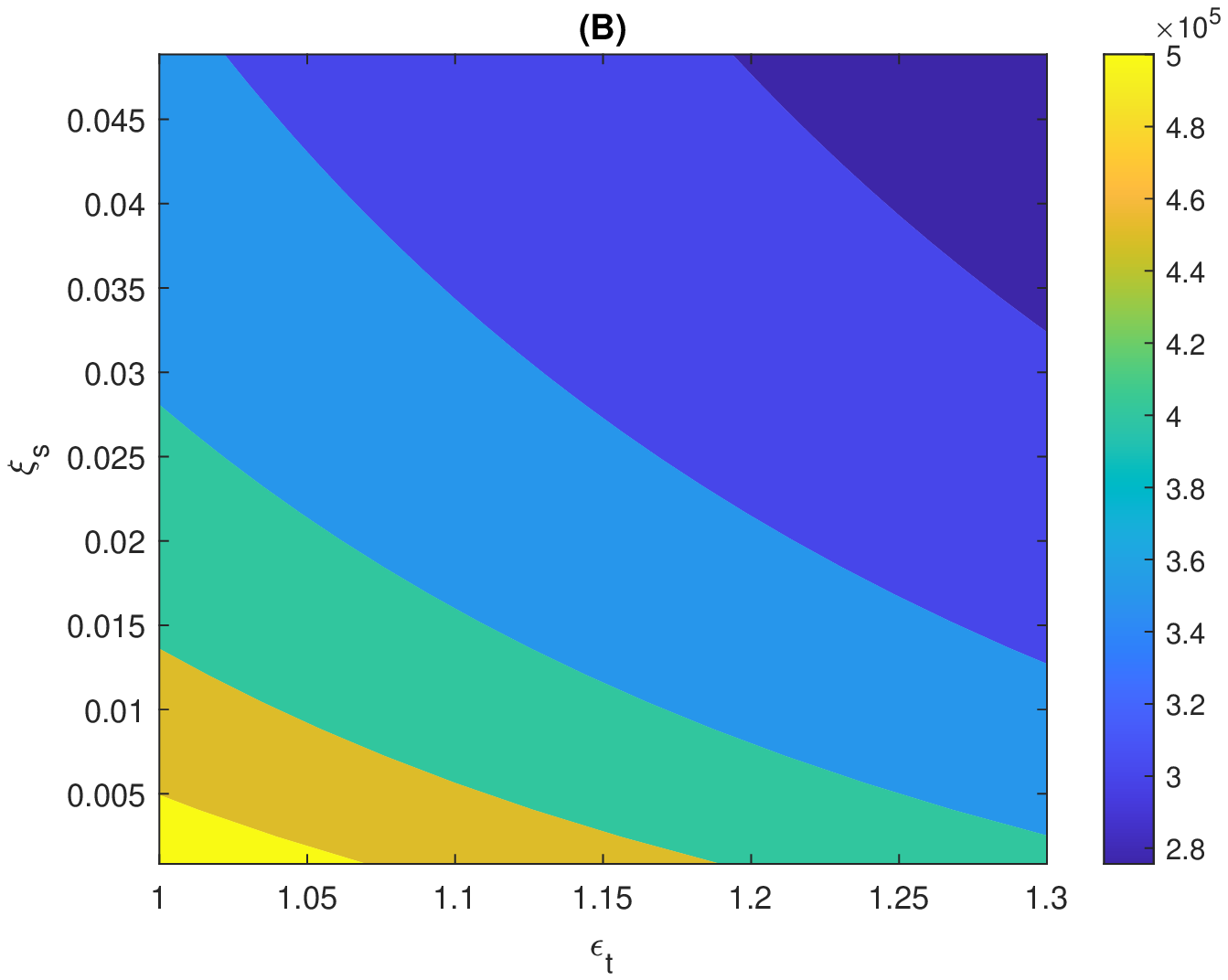}
	\includegraphics[width=0.45\textwidth]{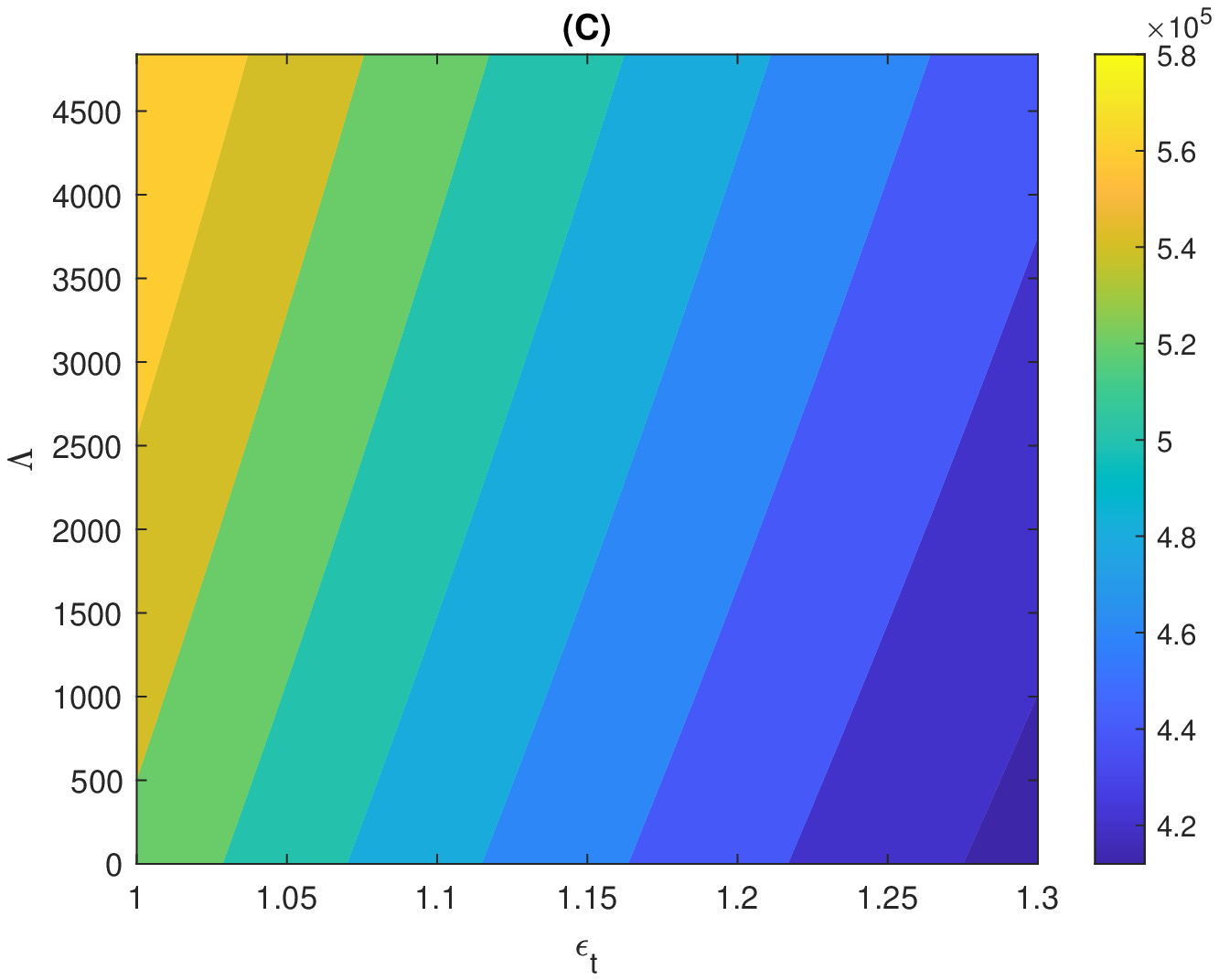}
	\includegraphics[width=0.45\textwidth]{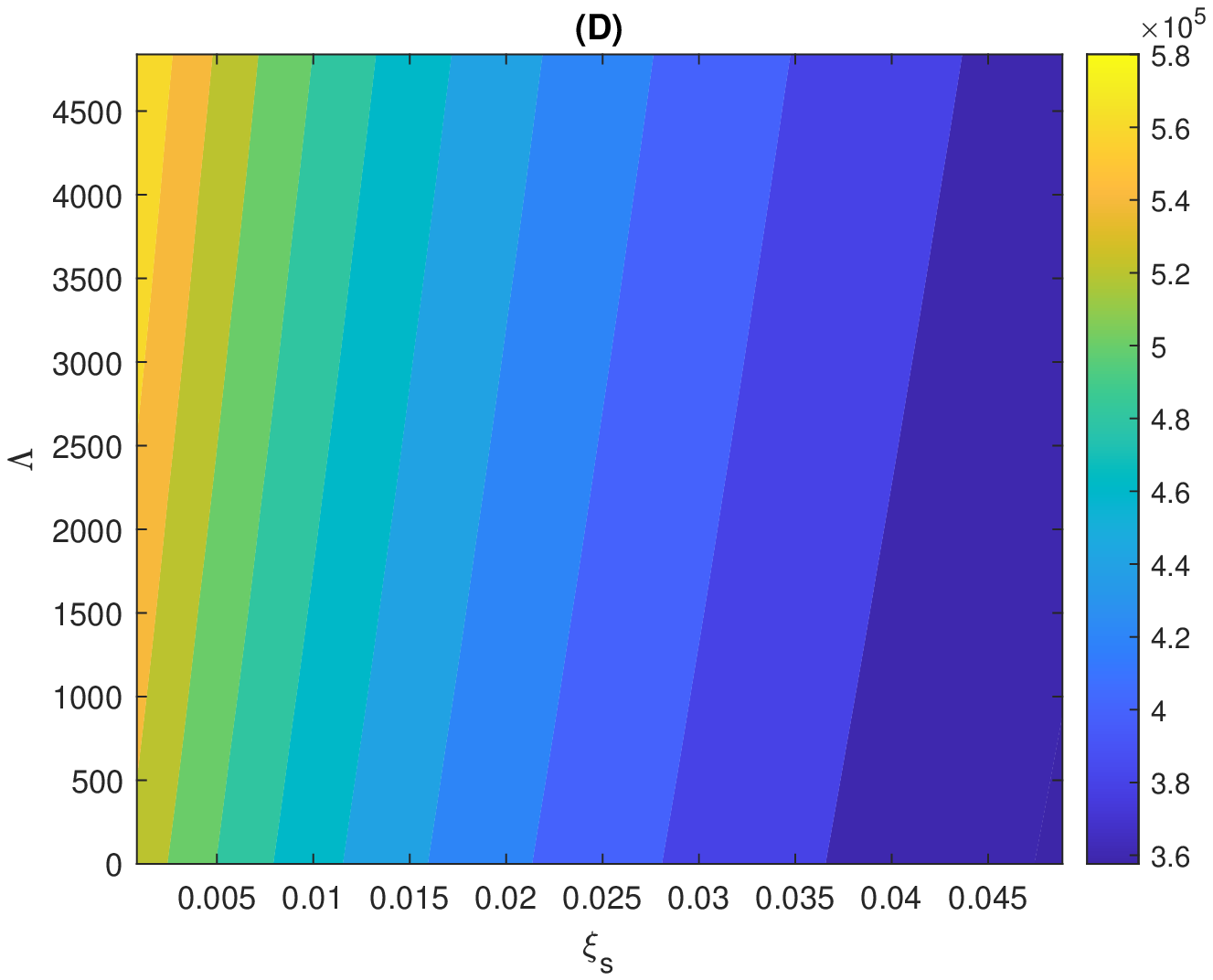}
	\caption{Combination of control strategies and immigration of infectives in Chennai. (A)  efficacy of face mask usage  $ \textendash$ community-wide compliance in face mask usage ($\epsilon_m-c_m$), (B) increase in recovery rate of hospitalized patients $ \textendash$ vaccination rate of susceptible individuals ($\epsilon_t-\xi_s$), (C) increase in recovery rate of hospitalized patients $ \textendash$ immigration of infectives ($\epsilon_t-\Lambda$), (D) vaccination rate of susceptible individuals $ \textendash$ immigration of infectives ($\epsilon_s-\Lambda$).}
	\label{Fig:two_controls_Chennai}
\end{figure}

Three or more control interventions and the immigration of infectives are now applied to examine their combined impact on the COVID-19 cases in the three months projection period. We calculate the percentage reduction in notified COVID-19 patients and hospitalized COVID-19 patients using the formula \eqref{percent_reduction}. The fixed parameters and initial conditions are taken from Tables \ref{tab:mod1} and \ref{tab:ICs_blore_chennai}. Four different combination strategies are investigated namely, \textit{strategy I}: face mask efficacy - community level usage of face masks - vaccination of susceptibles ($\epsilon_m$-$c_m$-$\xi_s$), \textit{strategy II}: face mask efficacy - community level usage of face masks - vaccination of susceptibles - increase in treatment rate ($\epsilon_m$-$c_m$-$\xi_s$-$\epsilon_t$), \textit{strategy III}: face mask efficacy - community level usage of face masks - vaccination of susceptibles - increase in treatment rate - immigration of un-notified infectives ($\epsilon_m$-$c_m$-$\xi_s$-$\epsilon_t$-$\Lambda$) and \textit{strategy IV}: combination of all control interventions with immigration of un-notified infectives ($\epsilon_m$-$c_m$-$\xi_s$-$\epsilon_t$-$\Lambda$-$\xi_r$). The parameter values are chosen for feasibility of the strategies i.e., we do not consider hypothetically high values of the parameters. The percentage reduction of $I_n$ and $I_h$ for Bangalore urban and Chennai are reported in Table \ref{percent_reduction_multiple_control}. 

\begin{table}
	\centering \caption{Percentage reduction in total number of notified and hospitalized COVID-19 patients for different strategies.}\label{percent_reduction_multiple_control}
	\hskip-0.9cm\begin{tabular}{|c|p{4cm}|c|c|c|c|}
		\hline
		&\multirow{2}{*}{Parameter values} & \multicolumn{2}{c|}{Bangalore urban} & \multicolumn{2}{c|}{Chennai}\\ \cline{3-6}
		& & Reduction in $I_n$ & Reduction in $I_h$  & Reduction in $I_n$ & Reduction in $I_h$ \\ \hline
		&$\epsilon_m$ = 0.38, $c_m$ = 0.25, $\xi_s$ = 0.005  & 16.92 & 15.12 & 13.66 & 12.80 \\ \cline{2-6}
		\multirow{2}{*}{\rotatebox[origin=c]{90}{Strategy-I}}&$\epsilon_m$ = 0.38, $c_m$ = 0.50, $\xi_s$ = 0.005	& 25.71 & 23.02 & 21.18 & 19.87 \\ \cline{2-6}
		&$\epsilon_m$ = 0.55, $c_m$ = 0.25, $\xi_s$ = 0.005  & 20.92 & 18.72 & 17.03 & 15.98 \\ \hline
		&$\epsilon_m$ = 0.38, $c_m$ = 0.25, $\xi_s$ = 0.005, $\epsilon_t$ = 1.1  & 16.92 & 22.58 & 13.65 & 20.49 \\ \cline{2-6}
		\multirow{2}{*}{\rotatebox[origin=c]{90}{Strategy-II}}&$\epsilon_m$ = 0.38, $c_m$ = 0.50, $\xi_s$ = 0.005, $\epsilon_t$ = 1.1  & 25.71 & 29.77 & 21.18 & 26.92 \\ \cline{2-6}
		&$\epsilon_m$ = 0.55, $c_m$ = 0.25, $\xi_s$ = 0.005, $\epsilon_t$ = 1.1  & 20.92 & 25.85 & 17.03 & 23.37 \\ \hline
		&$\epsilon_m$ = 0.38, $c_m$ = 0.25, $\xi_s$ = 0.005, $\epsilon_t$ = 1.1, $\Lambda$ = 100  & 16.76 & 22.33 & 13.50 & 20.35 \\  \cline{2-6}
		&$\epsilon_m$ = 0.38, $c_m$ = 0.50, $\xi_s$ = 0.005, $\epsilon_t$ = 1.1, $\Lambda$ = 100  & 25.57 & 29.53 & 21.03 & 26.78 \\  \cline{2-6}
		\multirow{4}{*}{\rotatebox[origin=c]{90}{Strategy-III}}&$\epsilon_m$ = 0.55, $c_m$ = 0.25, $\xi_s$ = 0.005, $\epsilon_t$ = 1.1, $\Lambda$ = 100 & 20.77 & 25.61 & 16.88 & 23.24 \\  \cline{2-6}
		&$\epsilon_m$ = 0.38, $c_m$ = 0.25, $\xi_s$ = 0.005, $\epsilon_t$ = 1.1, $\Lambda$ = 1000  & 15.38 & 20.11 & 12.15 & 19.11 \\  \cline{2-6}
		&$\epsilon_m$ = 0.38, $c_m$ = 0.50, $\xi_s$ = 0.005, $\epsilon_t$ = 1.1, $\Lambda$ = 1000  & 24.28 & 27.38 & 19.75 & 25.60 \\  \cline{2-6}
		&$\epsilon_m$ = 0.55, $c_m$ = 0.25, $\xi_s$ = 0.005, $\epsilon_t$ = 1.1, $\Lambda$ = 1000 & 19.42 & 23.42 & 15.56 & 22.02 \\ \hline
		\multirow{2}{*}{\rotatebox[origin=c]{90}{Strategy-IV~~~~}}&$\epsilon_m$ = 0.38, $c_m$ = 0.50, $\xi_s$ = 0.005, $\xi_r$ = 0.01, $\epsilon_t$ = 1.1, $\Lambda$ = 100  & 25.99 & 29.87 & 21.55 & 27.23 \\ \cline{2-6}
		& $\epsilon_m$ = 0.38, $c_m$ = 0.50, $\xi_s$ = 0.005, $\xi_r$ = 0.01, $\epsilon_t$ = 1.1, $\Lambda$ = 1000  & 24.78 & 27.78 & 20.36 & 26.12 \\ \hline
	\end{tabular}
\end{table}

From the combinations of strategy I and strategy II, we observe that the second case of strategy II has most the positive reduction for both the districts. However, the travel restrictions are being eased and more people tend to come to these districts. Thus, the immigration of un-notified infectives is also evident and hence it is necessary to consider non-zero values of $\Lambda$. It can be observed that strategy IV (first case) with specified control parameter values and low immigration has the most fruitful results on both the locations (marked bold in Table \ref{percent_reduction_multiple_control}). Alternatively, the second case of strategy II show a similar case reduction in both the locations. This strategy does not include immigration and therefore is not feasible. It can be seen that the second case of strategy III also shows competitive results in terms of percentage reduction of both populations. This strategy is feasible and can be implemented with less effort by the policy makers.

\section{Results and conclusions}\label{conclusion}
An SEIR-type compartmental model for the transmission dynamics of COVID-19 outbreak incorporating free virus concentration of the environment is considered in this paper. The model assumes standard incidence function while COVID-19 transmission by un-notified and notified individuals. Mass action incidence is considered for the airborne transmission of the virus. The mathematical model consists of a seven-dimensional system of ordinary differential equations. The model is mathematically analysed to obtain insight into the long-term dynamic features of the disease. Basic properties of the model are studied by finding positivity and boundedness of the solutions when the initial conditions are taken non-negative for all the state variables. The disease-free equilibrium (DFE) of the model is found to be unique and the basic reproduction number is found using the next-generation matrix method. The model has a locally stable DFE whenever the basic reproduction number is less than unity otherwise the DFE becomes unstable. Moreover, the global stability of the DFE is guaranteed when $\mathscr{R}_0<1$. This can be inferred that COVID-19 will be mitigated from the community if the corresponding parameters are in the range to ensure $\mathscr{R}_0<1$. The existence of the endemic equilibrium is investigated and it is found that at most two endemic equilibria may exist for the system. Forward transcritical bifurcation is established both analytically and numerically between DFE and one of the endemic equilibrium (see Fig. \ref{forward_bif}). Further, it is numerically shown that all the infected compartments tend to vanish in the long run when $\mathscr{R}_0<1$ and approaches a non-zero equilibrium when $\mathscr{R}_0>1$ (shown in Fig. \ref{time_series}). 

We calibrated the proposed model parameters to fit daily mortality data of two districts of India during the period of March $1^{st}$, 2021 to May $31^{st}$, 2021. From the fitting of the model output to mortality data, it can be seen that the model output closely anticipates the real data for both the districts. Afterwards, we investigate the impacts of control strategies and immigration of un-notified persons on the number of hospitalized patients. We concentrate on both pharmaceutical and non-pharmaceutical control measures as they are being implemented simultaneously to combat the deadly disease. Two pharmaceutical measures namely, treatment of hospitalized persons and vaccination of susceptibles as well as recovered persons are considered. On the other hand, face mask efficacy and community-wide compliance of face masks are considered as non-pharmaceutical control measures. Along with these interventions, we also study the effect of immigrating un-notified persons from other communities. Individual controls and immigration of infectives are initially projected to understand their sole impact on the hospitalized cases. It is found that treatment of hospitalized patients and vaccination of susceptible impacts has a significant impact on the reduction of hospitalized cases (see Fig. \ref{Fig:single_control_Bangalore} and Fig. \ref{Fig:single_control_chennai}). Therefore, pharmaceutical measures are more reliable than non-pharmaceutical controls. Additionally, immigration of un-notified infected persons may drive the hospitalized cases to higher prevalence. We calculated the percentage reduction in total hospitalized patients to get a quantitative idea of the control measures and immigration. It is reinforced that face mask related controls have a lower impact than pharmaceutical measures like treatment and vaccination. The immigration of un-notified patients has a negative impact on the total hospitalized cases and thus this should be kept as low as possible. Low immigration of infectives can be maintained by restricting immigration from highly affected neighbouring districts. Further, we investigate the combined effect of two strategies on the total number of hospitalized patients for both districts. We found that the combination of treatment and vaccination of susceptible individuals will be very effective to reduce the disease. Finally, we compare three or more interventions by calculating the percentage reduction of notified and hospitalized patients with different levels of feasible control measures and immigration rates. We examine strategy I -- strategy IV (as defined in the previous section) and find that combination of all controls with low immigration will have the best effect on the reduction of the disease. However, the second case of strategy III also show competitive results in terms of percentage reduction of both populations. This strategy differs from the strategy IV by the vaccination of recovered people. Therefore, it can be concluded that when vaccine crisis is there, the public health authorities may not vaccinate the recovered people. In other words, the people who were tested COVID-19 positive in past and are now healthy (they may have developed natural antibodies against SARS-CoV-2) may be vaccinated later when vaccines are available. Low immigration of un-notified patients can be regulated by restricting the movement of people from high prevalence adjacent districts. This analysis can be made with data from other profoundly affected districts with proper parametrization.

In summary, the results indicate that vaccination of susceptible individuals and treatment of hospitalized patients are very crucial to control COVID-19 in the two locations. In the three months short-term period, vaccination of susceptible people should be prioritized over vaccination of recovered people for a better outcome. Whenever vaccine crisis is there, governments may avoid the vaccination of people who were tested COVID-19 positive and are now healthy. However, increased quality and quantity of mask use are also helpful. Along with face mask use, treatment of hospitalized patients and vaccination of susceptibles, immigration should be allowed in a supervised manner so that economy of the overall society remain healthy.

\section*{Acknowledgements} 
Research of Indrajit Ghosh is supported by National Board for Higher Mathematics (NBHM) postdoctoral fellowship (Ref. No: 0204/3/2020/R \& D-II/2458), Department of Atomic Energy, Government of India. 

\section*{Conflict of interest}
The authors declare that they have no conflict of interest.
%%%%%%%%%%%%%%%%%%%%%%%%%%%%%%%%%%%%%%%%%%%%%%%%%%%%%%%%%%%%%%%%

\bibliographystyle{plain}
\biboptions{square}
\bibliography{bib_covid}

\begin{thebibliography}{10}

\bibitem{Worldometer2021covid}
{Coronavirus} worldometer.
\newblock \url{https://www.worldometers.info/coronavirus/}, 2021.
\newblock Retrieved : 2021-06-04.

\bibitem{covid19india}
covid19india.
\newblock \url{https://www.covid19india.org/}, 2021.
\newblock Retrieved : 2021-06-09.

\bibitem{Worldometer2021demo}
{Demographics} worldometer.
\newblock \url{https://www.worldometers.info/demographics/india-demographics/},
  2021.
\newblock Retrieved : 2021-06-04.

\bibitem{Bangalore2021pop}
{Macrotrends} bangalore population.
\newblock \url{https://www.macrotrends.net/cities/21176/bangalore/population},
  2021.
\newblock Retrieved : 2021-06-09.

\bibitem{Chennai2021pop}
{Macrotrends} chennai population.
\newblock \url{https://www.macrotrends.net/cities/21321/madras/population},
  2021.
\newblock Retrieved : 2021-06-09.

\bibitem{addleman2021mitigating}
Sarah Addleman, Victor Leung, Leyla Asadi, Abdu Sharkawy, and Jennifer
  McDonald.
\newblock Mitigating airborne transmission of sars-cov-2.
\newblock {\em CMAJ}, 193(26):E1010--E1011, 2021.

\bibitem{andrews2020first}
MA~Andrews, Binu Areekal, KR~Rajesh, Jijith Krishnan, R~Suryakala, Biju
  Krishnan, CP~Muraly, and PV~Santhosh.
\newblock First confirmed case of covid-19 infection in india: A case report.
\newblock {\em The Indian Journal of Medical Research}, 151(5):490, 2020.

\bibitem{bere20212dg}
Jyostna Bere, Naga~Dharant Jonnalagadda, Lavanya Kappari, Jyoshika Karangula,
  Narender Boggula, and Vedhavathi Kappari.
\newblock 2-deoxy-d-glucose: an update review.
\newblock {\em Journal For Innovative Development in Pharmaceutical and
  Technical Science (JIDPTS)}, 4(5):68--78, 2021.

\bibitem{castillo2002computation}
Carlos Castillo-Chavez, Zhilan Feng, Wenzhang Huang, et~al.
\newblock On the computation of $\mathcal{R}_0$ and its role in global
  stability, in mathematical approaches for emerging and reemerging infectious
  diseases: an introduction (minneapolis, mn, 1999).
\newblock {\em IMA Volumes in Mathematics and Its Applications, New York,
  Springer}, 125:229--250, 2002.

\bibitem{castillo2004dynamical}
Carlos Castillo-Chavez and Baojun Song.
\newblock Dynamical models of tuberculosis and their applications.
\newblock {\em Mathematical Biosciences \& Engineering}, 1(2):361--404, 2004.

\bibitem{choi2020optimal}
Wongyeong Choi and Eunha Shim.
\newblock Optimal strategies for vaccination and social distancing in a
  game-theoretic epidemiologic model.
\newblock {\em Journal of Theoretical Biology}, 505:110422, 2020.

\bibitem{diekmann1990definition}
Odo Diekmann, Johan Andre~Peter Heesterbeek, and Johan~AJ Metz.
\newblock On the definition and the computation of the basic reproduction ratio
  $r_0$ in models for infectious diseases in heterogeneous populations.
\newblock {\em Journal of Mathematical Biology}, 28(4):365--382, 1990.

\bibitem{editorial2020covid}
Editorial.
\newblock Covid-19 transmission—up in the air.
\newblock {\em The Lancet. Respiratory Medicine}, 8(12):1159, 2020.

\bibitem{foy2021comparing}
Brody~H Foy, Brian Wahl, Kayur Mehta, Anita Shet, Gautam~I Menon, and Carl
  Britto.
\newblock Comparing covid-19 vaccine allocation strategies in india: A
  mathematical modelling study.
\newblock {\em International Journal of Infectious Diseases}, 103:431--438,
  2021.

\bibitem{ghosh2021modeling}
Indrajit Ghosh and Maia Martcheva.
\newblock Modeling the effects of prosocial awareness on covid-19 dynamics:
  Case studies on colombia and india.
\newblock {\em Nonlinear Dynamics}, 104:4681--4700, 2021.

\bibitem{greenhalgh2021ten}
Trisha Greenhalgh, Jose~L Jimenez, Kimberly~A Prather, Zeynep Tufekci, David
  Fisman, and Robert Schooley.
\newblock Ten scientific reasons in support of airborne transmission of
  sars-cov-2.
\newblock {\em The Lancet}, 397(10285):1603--1605, 2021.

\bibitem{kermack1927contribution}
William~Ogilvy Kermack and Anderson~G McKendrick.
\newblock A contribution to the mathematical theory of epidemics.
\newblock {\em Proceedings of the Royal Society A: Mathematical Physical and
  Engineering Sciences}, 115(772):700--721, 1927.

\bibitem{lau2021evaluating}
Hien Lau, Tanja Khosrawipour, Piotr Kocbach, Hirohito Ichii, Jacek Bania, and
  Veria Khosrawipour.
\newblock Evaluating the massive underreporting and undertesting of covid-19
  cases in multiple global epicenters.
\newblock {\em Pulmonology}, 27(2):110--115, 2021.

\bibitem{li2020early}
Qun Li, Xuhua Guan, Peng Wu, Xiaoye Wang, Lei Zhou, Yeqing Tong, Ruiqi Ren,
  Kathy~SM Leung, Eric~HY Lau, Jessica~Y Wong, et~al.
\newblock Early transmission dynamics in wuhan, china, of novel
  coronavirus--infected pneumonia.
\newblock {\em New England Journal of Medicine}, 382:1199--1207, 2020.

\bibitem{linton2020incubation}
Natalie~M Linton, Tetsuro Kobayashi, Yichi Yang, Katsuma Hayashi, Andrei~R
  Akhmetzhanov, Sung-mok Jung, Baoyin Yuan, Ryo Kinoshita, and Hiroshi
  Nishiura.
\newblock Incubation period and other epidemiological characteristics of 2019
  novel coronavirus infections with right truncation: a statistical analysis of
  publicly available case data.
\newblock {\em Journal of Clinical Medicine}, 9(2):538, 2020.

\bibitem{liu2020viral}
Yang Liu, Li-Meng Yan, Lagen Wan, Tian-Xin Xiang, Aiping Le, Jia-Ming Liu,
  Malik Peiris, Leo~LM Poon, and Wei Zhang.
\newblock Viral dynamics in mild and severe cases of covid-19.
\newblock {\em The Lancet Infectious Diseases}, 20(6):656--657, 2020.

\bibitem{lopez2020end}
Leonardo L{\'o}pez and Xavier Rod{\'o}.
\newblock The end of social confinement and covid-19 re-emergence risk.
\newblock {\em Nature Human Behaviour}, 4(7):746--755, 2020.

\bibitem{mandal2020prudent}
Sandip Mandal, Tarun Bhatnagar, Nimalan Arinaminpathy, Anup Agarwal, Amartya
  Chowdhury, Manoj Murhekar, Raman~R Gangakhedkar, and Swarup Sarkar.
\newblock Prudent public health intervention strategies to control the
  coronavirus disease 2019 transmission in india: A mathematical model-based
  approach.
\newblock {\em The Indian Journal of Medical Research}, 151(2-3):190, 2020.

\bibitem{martcheva2015introduction}
Maia Martcheva.
\newblock {\em An Introduction to Mathematical Epidemiology}, volume~61.
\newblock Springer, New York, 2015.

\bibitem{mwalili2020seir}
Samuel Mwalili, Mark Kimathi, Viona Ojiambo, Duncan Gathungu, and Rachel Mbogo.
\newblock Seir model for covid-19 dynamics incorporating the environment and
  social distancing.
\newblock {\em BMC Research Notes}, 13(1):1--5, 2020.

\bibitem{ngonghala2020could}
Calistus~N Ngonghala, Enahoro~A Iboi, and Abba~B Gumel.
\newblock Could masks curtail the post-lockdown resurgence of covid-19 in the
  us?
\newblock {\em Mathematical Biosciences}, 329:108452, 2020.

\bibitem{dcgi2021pop}
PIB.
\newblock Dcgi approves anti-covid drug developed by drdo for emergency use.
\newblock \url{https://pib.gov.in/PressReleasePage.aspx?PRID=1717007}, 2021.
\newblock Press Information Bureau, Government of India. 2021-05-08 Retrieved :
  2021-07-04.

\bibitem{rafiq2020evaluation}
Danish Rafiq, Suhail~Ahmad Suhail, and Mohammad~Abid Bazaz.
\newblock Evaluation and prediction of covid-19 in india: A case study of worst
  hit states.
\newblock {\em Chaos, Solitons \& Fractals}, 139:110014, 2020.

\bibitem{saldana2020modeling}
Fernando Salda{\~n}a, Hugo Flores-Arguedas, Jos{\'e}~Ariel
  Camacho-Guti{\'e}rrez, and Ignacio Barradas.
\newblock Modeling the transmission dynamics and the impact of the control
  interventions for the covid-19 epidemic outbreak.
\newblock {\em Mathematical Biosciences \& Engineering}, 17(4):4165--4183,
  2020.

\bibitem{samal2021anti}
Kailash~Chandra Samal, Bhagyalaxmi Panda, and Laxmipreeya Behera.
\newblock Anti-covid drug: 2-deoxy-d-glucose and its mechanism of action.
\newblock {\em Biotica Research Today}, 3(5):345--347, 2021.

\bibitem{sardar2020assessment}
Tridip Sardar, Sk~Shahid Nadim, Sourav Rana, and Joydev Chattopadhyay.
\newblock Assessment of lockdown effect in some states and overall india: A
  predictive mathematical study on covid-19 outbreak.
\newblock {\em Chaos, Solitons \& Fractals}, 139:110078, 2020.

\bibitem{sariol2020lessons}
Alan Sariol and Stanley Perlman.
\newblock Lessons for covid-19 immunity from other coronavirus infections.
\newblock {\em Immunity}, 53(2):248--263, 2020.

\bibitem{sarkar2020modeling}
Kankan Sarkar, Subhas Khajanchi, and Juan~J Nieto.
\newblock Modeling and forecasting the covid-19 pandemic in india.
\newblock {\em Chaos, Solitons \& Fractals}, 139:110049, 2020.

\bibitem{senapati2021impact}
Abhishek Senapati, Sourav Rana, Tamalendu Das, and Joydev Chattopadhyay.
\newblock Impact of intervention on the spread of covid-19 in india: A model
  based study.
\newblock {\em Journal of Theoretical Biology}, 523:110711, 2021.

\bibitem{sharma2020efficacy}
Suresh~K Sharma, Mayank Mishra, and Shiv~K Mudgal.
\newblock Efficacy of cloth face mask in prevention of novel coronavirus
  infection transmission: A systematic review and meta-analysis.
\newblock {\em Journal of Education and Health Promotion}, 9:192, 2020.

\bibitem{sharun2021india}
Khan Sharun and Kuldeep Dhama.
\newblock India’s role in covid-19 vaccine diplomacy.
\newblock {\em Journal of Travel Medicine}, 2021.

\bibitem{soin2021tocilizumab}
Arvinder~S Soin, Kuldeep Kumar, Narendra~S Choudhary, Pooja Sharma, Yatin
  Mehta, Sushila Kataria, Deepak Govil, Vikas Deswal, Dhruva Chaudhry,
  Pawan~Kumar Singh, et~al.
\newblock Tocilizumab plus standard care versus standard care in patients in
  india with moderate to severe covid-19-associated cytokine release syndrome
  (covintoc): an open-label, multicentre, randomised, controlled, phase 3
  trial.
\newblock {\em The Lancet Respiratory Medicine}, 9(5):511--521, 2021.

\bibitem{stutt2020modelling}
Richard~OJH Stutt, Renata Retkute, Michael Bradley, Christopher~A Gilligan, and
  John Colvin.
\newblock A modelling framework to assess the likely effectiveness of facemasks
  in combination with ‘lock-down’ in managing the covid-19 pandemic.
\newblock {\em Proceedings of the Royal Society A: Mathematical Physical and
  Engineering Sciences}, 476(2238):20200376, 2020.

\bibitem{tindale2020transmission}
Lauren Tindale, Michelle Coombe, Jessica~E Stockdale, Emma Garlock, Wing
  Yin~Venus Lau, Manu Saraswat, Yen-Hsiang~Brian Lee, Louxin Zhang, Dongxuan
  Chen, Jacco Wallinga, et~al.
\newblock Transmission interval estimates suggest pre-symptomatic spread of
  covid-19.
\newblock {\em MedRxiv}, 2020.

\bibitem{van2002reproduction}
Pauline {van den Driessche} and James Watmough.
\newblock Reproduction numbers and sub-threshold endemic equilibria for
  compartmental models of disease transmission.
\newblock {\em Mathematical Biosciences}, 180(1-2):29--48, 2002.

\bibitem{van2008professional}
Marianne {van der Sande}, Peter Teunis, and Rob Sabel.
\newblock Professional and home-made face masks reduce exposure to respiratory
  infections among the general population.
\newblock {\em PloS One}, 3(7):e2618, 2008.

\bibitem{van2020aerosol}
Neeltje {van Doremalen}, Trenton Bushmaker, Dylan~H Morris, Myndi~G Holbrook,
  Amandine Gamble, Brandi~N Williamson, Azaibi Tamin, Jennifer~L Harcourt,
  Natalie~J Thornburg, Susan~I Gerber, et~al.
\newblock Aerosol and surface stability of sars-cov-2 as compared with
  sars-cov-1.
\newblock {\em New England Journal of Medicine}, 382(16):1564--1567, 2020.

\bibitem{who2020coronavirus}
WHO.
\newblock Coronavirus disease 2019 (covid-19): situation report, 73.
\newblock \url{https://apps.who.int/iris/handle/10665/331686}, 2020.

\bibitem{who2020rolling}
WHO.
\newblock Rolling updates on coronavirus disease (covid-19): Who characterizes
  covid-19 as a pandemic.
\newblock
  \url{https://www.who.int/emergencies/diseases/novel-coronavirus-2019/events-as-they-happen},
  2020.

\bibitem{woelfel2020clinical}
Roman Woelfel, Victor~Max Corman, Wolfgang Guggemos, Michael Seilmaier, Sabine
  Zange, Marcel~A Mueller, Daniela Niemeyer, Patrick Vollmar, Camilla Rothe,
  Michael Hoelscher, et~al.
\newblock Clinical presentation and virological assessment of hospitalized
  cases of coronavirus disease 2019 in a travel-associated transmission
  cluster.
\newblock {\em MedRxiv}, 2020.

\bibitem{wu2020nowcasting}
Joseph~T Wu, Kathy Leung, and Gabriel~M Leung.
\newblock Nowcasting and forecasting the potential domestic and international
  spread of the 2019-ncov outbreak originating in wuhan, china: a modelling
  study.
\newblock {\em The Lancet}, 395(10225):689--697, 2020.

\bibitem{wu2020characteristics}
Zunyou Wu and Jennifer~M McGoogan.
\newblock Characteristics of and important lessons from the coronavirus disease
  2019 (covid-19) outbreak in china: summary of a report of 72 314 cases from
  the chinese center for disease control and prevention.
\newblock {\em JAMA}, 323(13):1239--1242, 2020.

\bibitem{xu2020clinical}
Xiao-Wei Xu, Xiao-Xin Wu, Xian-Gao Jiang, Kai-Jin Xu, Ling-Jun Ying, Chun-Lian
  Ma, Shi-Bo Li, Hua-Ying Wang, Sheng Zhang, Hai-Nv Gao, et~al.
\newblock Clinical findings in a group of patients infected with the 2019 novel
  coronavirus (sars-cov-2) outside of wuhan, china: retrospective case series.
\newblock {\em BMJ}, 368:m606, 2020.

\bibitem{yan2020impact}
Qinling Yan, Yingling Tang, Dingding Yan, Jiaying Wang, Linqian Yang, Xinpei
  Yang, and Sanyi Tang.
\newblock Impact of media reports on the early spread of covid-19 epidemic.
\newblock {\em Journal of Theoretical Biology}, 502:110385, 2020.

\bibitem{yang2020epidemiological}
Penghui Yang, Yibo Ding, Zhe Xu, Rui Pu, Ping Li, Jin Yan, Jiluo Liu, Fanping
  Meng, Lei Huang, Lei Shi, et~al.
\newblock Epidemiological and clinical features of covid-19 patients with and
  without pneumonia in beijing, china.
\newblock {\em MedRxiv}, 2020.

\end{thebibliography}

\end{document}